\documentclass[sigconf, nonacm]{acmart}

\usepackage{balance}

\newcommand\vldbdoi{10.14778/3430915.3430931}
\newcommand\vldbpages{418 - 430}
\newcommand\vldbvolume{14}
\newcommand\vldbissue{3}
\newcommand\vldbyear{2020}

\newcommand\vldbavailabilityurl{}
\newcommand\vldbpagestyle{empty}

\usepackage{cleveref}
\usepackage[mathscr]{euscript}
\usepackage{tikz}
\usepackage{tikz-qtree}
\usepackage{subcaption}
\usepackage{paralist}
\usetikzlibrary{arrows.meta}
\usetikzlibrary{fit,shapes,trees,shapes.geometric}
\usetikzlibrary{patterns,decorations.pathreplacing,calc}
\usetikzlibrary{arrows}
\usepackage[linesnumbered,algoruled, lined, noend]{algorithm2e}

\usepackage[normalem]{ulem}

\newcommand{\eat}[1]{}

\newcommand{\shaleen}[2]{}%{\acomment{#1}{#2}{magenta}{Shaleen}}
\newcommand{\anja}[2]{}%{\acomment{#1}{#2}{green}{Anja}}
\newcommand{\stratis}[2]{}%{\acomment{#1}{#2}{red}{Stratis}}
\newcommand{\cut}[1]{}   % remove things inside the cut
\newcommand{\framework}{\mbox{\textsc{GSum}}\xspace}

\newenvironment{packed_item}{
	\begin{itemize}
		\setlength{\itemsep}{1pt}
		\setlength{\parskip}{0pt}
		\setlength{\parsep}{0pt}
	}
	{\end{itemize}}

\newenvironment{packed_enum}{
	\begin{enumerate}
		\setlength{\itemsep}{1pt}
		\setlength{\parskip}{0pt}
		\setlength{\parsep}{0pt}
	}
	{\end{enumerate}}

\newcommand{\ie}{{\em i.e.}\xspace}

\newcommand{\introparagraph}[1]{\noindent {\bf \em #1.}}  % define own new subsection type: noindent, bold (textsc)

\usepackage{aliascnt} 

\newtheorem{definition}{Definition}

\newtheorem{proposition}{Proposition}

\newtheorem{lemma}{Lemma}
\newtheorem{example}{Example}

\newtheorem{theorem}{Theorem}

\providecommand{\norm}[1]{{\lVert#1\rVert}}

\providecommand{\bF}[0]{\mathbf{F}}
\providecommand{\bS}[0]{\mathbf{S}}

\providecommand{\bW}[0]{\mathbf{W}}

\providecommand{\mult}[3]{\textsf{m}_{#1}({#2},{#3})}

\providecommand{\domain}[0]{\mathsf{dom}}

\providecommand{\bq}[0]{\mathbf{q}}

\providecommand{\td}{d}
\let\oldlangle\langle
\renewcommand{\langle}{{\color{purple} \oldlangle}}
\let\oldrangle\rangle
\renewcommand{\rangle}{{\color{purple} \oldrangle}}

\definecolor{darkcyan}{rgb}{0.0, 0.35, 0.85}
\newcommand{\rone}[1]{{\color{black} #1}}
\definecolor{burntorange}{rgb}{1, 0.033, 0.24}
\newcommand{\rtwo}[1]{{\color{black} #1}}
\definecolor{forestgreen}{rgb}{0.13, 0.55, 0.13}
\newcommand{\rthree}[1]{{\color{black} #1}}

\newcommand{\johnc}[1]{{\color{black} #1}}

\begin{document}
	\title{Comprehensive and Efficient Workload Compression}
	\author{Shaleen Deep{$^\dagger$}, Anja Gruenheid{$^\ddagger$}, Paraschos Koutris{$^\dagger$}, Jeffrey Naughton{$^\ddagger$}, Stratis Viglas{$^\ddagger$}}
	\affiliation{	
	\institution{$^\dagger$University of Wisconsin - Madison, $^\ddagger$Google Inc.} 
	}
	\email{ {shaleen, paris}@cs.wisc.edu} 
	\email{  {anjag, naughton, sviglas}@google.com}

	\begin{abstract}
This work studies the problem of constructing a representative workload from a given input \rtwo{analytical} query workload where the former serves as an approximation with guarantees of the latter.
We discuss our work in the context of workload analysis and monitoring.
As an example, evolving system usage patterns in a database system can cause load imbalance and performance regressions which can be controlled by monitoring system usage patterns, i.e.,~a representative workload, over time.
To construct such a workload in a principled manner, we formalize the notions of workload {\em representativity} and {\em coverage}.
These metrics capture the intuition that the distribution of features in a compressed workload should match a target distribution, increasing representativity, and include common queries as well as outliers, increasing coverage. 
We show that solving this problem optimally is computationally hard and present a novel greedy algorithm that provides approximation guarantees. 
We compare our techniques to established algorithms in this problem space such as sampling and clustering, and demonstrate advantages and key trade-offs.
\end{abstract}

	\maketitle
	
	%%% do not modify the following VLDB block %%
	%%% VLDB block start %%%
	\pagestyle{\vldbpagestyle}
	\begingroup\small\noindent\raggedright\textbf{PVLDB Reference Format:}\\
	Shaleen Deep, Anja Gruenheid, Paraschos Koutris, Jeffrey Naughton, Stratis Viglas. Comprehensive and Efficient Workload Compression. PVLDB, \vldbvolume(\vldbissue): \vldbpages, \vldbyear.\\
	\href{https://doi.org/\vldbdoi}{doi:\vldbdoi}
	\endgroup
	\begingroup
	\renewcommand\thefootnote{}\footnote{\noindent This work is licensed under the Creative Commons BY-NC-ND 4.0 International License. Visit \url{https://creativecommons.org/licenses/by-nc-nd/4.0/} to view a copy of this license. For any use beyond those covered by this license, obtain permission by emailing \href{mailto:info@vldb.org}{info@vldb.org}. Copyright is held by the owner/author(s). Publication rights licensed to the VLDB Endowment. \\
		\raggedright Proceedings of the VLDB Endowment, Vol. \vldbvolume, No. \vldbissue\ %
		ISSN 2150-8097. \\
		\href{https://doi.org/\vldbdoi}{doi:\vldbdoi} \\
	}\addtocounter{footnote}{-1}\endgroup
	%%% VLDB block end %%%
	
	%%% do not modify the following VLDB block %%
	%%% VLDB block start %%%
	\ifdefempty{\vldbavailabilityurl}{}{
		\vspace{.3cm}
		\begingroup\small\noindent\raggedright\textbf{PVLDB Artifact Availability:}\\
		The source code, data, and/or other artifacts have been made available at \url{\vldbavailabilityurl}.
		\endgroup
	}
	%%% VLDB block end %%%
	
	\section{Introduction}
\label{sec:intro}

Performance tuning has been at the core of database system development and deployment since its inception.
To facilitate effective system design and development, we need to understand how the system is used over time. 
For example, if a system is developed for transactional workloads but is increasingly used for analytical workloads, its usage patterns significantly shift, potentially resulting in performance regression.
The first step towards a holistic understanding of system usage is to perform an in-depth analysis of the query workloads it is serving. 
However, logs from production database systems are far too large to examine manually. 
To be able to identify key components of the workload, we propose to  create and monitor a subset of the input workload which {\em closely} represents the original workload.
To that end, our work presents a \rtwo{semi-supervised} framework to {\em compress} \rtwo{analytical} workloads.

\smallskip
\introparagraph{Problem Motivation}
The need for our work stems from the complexity of contemporary database deployments, which have scaled with the number of customers and the workloads to serve. 
In order to predict the performance of a \textsf{RDBMS} on a large workload, it is common to evaluate it on a {\em benchmark workload} that resembles the target workload. 
Historically, benchmarked workloads are either standardized (such as \textsf{TPC-H}~\cite{tpch}, \textsf{SSB}~\cite{o2007star}, \textsf{YCSB}~\cite{cooper2010benchmarking}, and \textsf{Wisconsin Benchmark}~\cite{dewitt1993wisconsin}) or created by domain experts who manually curate queries.
If we wish to construct a custom benchmark for every use case of each customer, the second solution becomes unsustainable.
Thus, automatic workload characterization and subsequently workload compression are a means to address this issue in an efficient and scalable way. 

\smallskip
\introparagraph{Approach}
With commercial deployments serving billions of queries per day, the size of a system's workload-to-analyze quickly escalates.
Instead of using every query of the workload, we propose to use a smaller sample of the workload while qualitatively not degrading the result of the process. We call this {\em workload compression} or {\em summarization}\footnote{We will be using the terms {\em workload compression} and {\em workload summarization} interchangeably, as we will for the terms {\em compressed workload} and {\em summary workload}}. 
Constructing a compressed workload is challenging for several reasons.
First, there is no universal set of goals to consider as representative for the workload: the output changes depending on the metric we optimize for. 
Second, the production workload to compress often does not fit any well-known statistical distribution, thereby making workload synthesis extremely challenging.
Lastly, there are a variety of variables to take into account in a real production deployment: different job sizes, a wide range of query run times, observable skew due to temporal or spatial locality, query complexity. It is, therefore, unclear what features are the salient ones when it comes to characterizing a workload adequately.
The two following examples illustrate the diverse characteristics of workload compression.

\begin{example}
	 \rone{Consider the developers of an application who use a query engine in production. The developers would like to create performance accountability of the query engine, i.e.,~they would like to create compliance benchmarks to track the query engine performance.} Suppose that application generates a workload mainly consists of short look-up queries but contains a handful of long-running queries, then both of these query types must be present in the benchmark.
	 Thus, the benchmark workload must have {\em high coverage} by including queries with differing run times \rone{to track the query engine performance on both types of queries.}
\end{example}

\begin{example}
	System performance can be tuned by recommending indexes to speed up query processing.
	However, the complexity of index recommendation grows quadratically with the workload size.
	Therefore, we would like to find a compressed workload that is {\em highly representative}, i.e.,~it has the same performance characteristics as the original input workload and use it for index recommendation instead.
\end{example}
%
%\rone{%
%\begin{example}
%	In production settings, it is common to monitor databases for atypical usage patterns that could indicate a serious bug or security threat. When	query logs are monitored, it is often done retrospectively, some hours after-the-fact. To support real-time monitoring, a key prerequisite is to quickly compute a compact subset of queries approximating the system’s typical workload that can be monitored and potentially help in proactive bug/threat discovery.
%\end{example}
%}

Finding a workload with both high representativity and high coverage is a challenging combinatorial problem for the following reasons: 
\begin{inparaenum}[(\itshape i\upshape)]
	\item {\em Defining metrics formally.} It is unclear what it means to have high coverage and representativity since these metrics are dependent on the application context.
    \item {\em High workload heterogeneity and variability.} While manual queries are smaller and easier to deal with, it is not uncommon that workload queries are machine-generated with more than 50 joins in production workloads.
    A good approximation of the input workload needs to contain all types of queries.
    \item {\em Increasing workload scale.}
    Production database systems routinely serve billions of queries every day, which makes any analysis challenging.
    Moreover, a significant fraction of these workloads is over ad hoc queries and tables rather than carefully designed schemas, making pattern mining over large, unstructured data sources even more difficult. 
\end{inparaenum}

\smallskip
\introparagraph{Prior work} Previous work on workload compression has used a variety of techniques ranging from random sampling and clustering to sophisticated ML models.
For instance,~\cite{chaudhuri2003primitives,chaudhuri2002compressing} employ clustering by defining a customized distance function for each application.
%The proposal is to define a new SQL operator called \textsc{Dominate} that performs summarization by using a custom objective function and then greedily choose queries.
%To ensure high coverage, the operator supports a limited set of constraints that the output must satisfy.
More recently,~\cite{jain2018query2vec} explores machine learning for workload compression. 
The insight here is to train a model specifically for SQL queries (similar to \textsc{Word2Vec}).
Most closely related to our setting is the query log summarization framework, \textsc{Ettu}~\cite{kul2016ettu, kul2016summarizing}.
\textsc{Ettu} summarizes query logs by parsing the syntax tree of queries and performing hierarchical clustering where the distance metric between queries is based on the number of common subexpressions. All of the above proposals have certain limitations.
The techniques in~\cite{chaudhuri2003primitives, chaudhuri2002compressing} are not scalable even for medium-sized workloads owing to high time complexity of $O(n^2)$ where $n$ is the input workload size.
While~\cite{jain2018query2vec} does not suffer from quadratic complexity, it requires expensive preprocessing to train the machine learning model.
%The approach described in~\cite{kul2016summarizing} exploits the observation that the number of important subexpressions across the SQL queries of a workload is limited, which helps the authors perform efficient clustering. 
The approach described in~\cite{kul2016summarizing} ignores query execution statistics, templated queries, and stored procedures.
Additionally, input workloads are often skewed in some way, and it is critical to ensure that the summary exhibits the same kind of characteristics as the input workload; this is not possible without some notion of representativity. 
\rone{Note that prior work, including our own, ignore any query execution impact of concurrently executed queries. 
For instance, they may compete for the same set of resources which in turn affects the performance.
Designing representative workloads with provable guarantees that also take such effects into account remains a challenging open problem.}

\smallskip
\introparagraph{Contributions and organization} In this work, we introduce a novel, generic framework for workload compression of analytical queries that applies to a wide variety of performance tuning tasks. We design, implement, and evaluate our workload compressor with robust guarantees for representativity and coverage. More specifically, we 
\begin{compactitem}
	\item formally define representativity and coverage to formulate workload compression as an optimization problem ( Section~\ref{sec:framework}) and propose a set of error metrics.
	\item prove that maximizing representativity is a hard problem even in restricted settings (Section~\ref{sec:hardness}).
	\item propose a novel objective that exploits submodularity to provide provable guarantees about the quality of a compressed workload (Section~\ref{sec:main}). Our algorithm allows for a smooth trade-off between representativity and coverage.
	\item apply optimizations to improve the performance of the algorithm and describe how submodularity can be exploited for efficient computation (Section~\ref{sec:parallel}). 
	\item evaluate our approach by comparing to sampling and clustering-based methods in Section~\ref{sec:eval}. We experimentally demonstrate that our framework and metrics are powerful enough by applying them to three practical use cases of $(i)$ schema index tuning; $(ii)$ materialized view recommendation; and $(iii)$ index and view recommendation. We show that our techniques require two orders of magnitude less time to create a compressed workload with better representativity and coverage as compared to clustering-based approaches. 
\end{compactitem}

%Lastly, we present a detailed comparison with related work in Section~\ref{sec:related} and conclude in Section~\ref{sec:conclusion}.
	\section{Problem Setting}
\label{sec:framework}

\rtwo{In this paper, we assume that a query log $L$ contains a finite collection of queries. 
For each query $q \in L$, we will create a representation that {\itshape featurizes} the query, such as finding predicates in the \texttt{\color{blue}WHERE} clause or table names present in the \texttt{\color{blue}FROM} clause. 
We assume that the universe of features in both a log and a query is enumerable and finite.
This requirement is essential in order to define appropriate metrics.} \rone{We also assume that the log contains execution statistics that can be looked over query-by-query. Such statistics are recorded by all DBMS (\cite{sqlserver} shows the statistics recorded by SQL Server).}
%We now introduce the notation and metrics used throughout this paper as well as the problem that we are solving.

\subsection{Notation}
\label{subsec:cq}

We define the input workload as a {\em multiset} $\mathbf{W} = \{ q_1, \dots, q_n\}$ that consists of $n$ queries. A workload is a multiset, since the same query may occur many times in the workload. 
\rtwo{Each $q_i$ is a log entry that contains the SQL text of the query and its execution information.}

\smallskip
\introparagraph{Features} 
We summarize a workload with respect to its feature set. 
We consider two types of features: 
\begin{compactitem}
    \item {\em Categorical features} ($\bF_{\text{categorical}}$): these features capture values derived from the syntax tree of a \textbf{SQL} query.
    \item {\em Numeric features} ($\bF_{\text{numeric}}$): these features capture numeric values that are derived from profiling statistics of the query.
\end{compactitem}

The feature set is defined as $\bF =  \bF_{\text{numeric}} \cup \bF_{\text{categorical}}$. We use $\domain(\bS, f)$ to denote the active domain of feature $f$ for some workload $\bS$ and use {\em token} to refer to feature values in the active domain. 

\smallskip
\introparagraph{Feature value multiplicity} To design a generic approach, we consider not only single-valued features but also multi-valued features. 
In other words, features of a query $\bq$ can have multiple domain values associated with it (that can also occur multiple times). 
\rtwo{For example, consider the function calls present in the \texttt{\color{blue}SELECT} clause.
Since a function call such as \texttt{\color{blue}SUM} can be present multiple times in the SQL statement, it is a multi-valued feature}.
To formally model multi-valued features, we represent the value of a feature $f$ of query $\bq$ as a multiset of tokens $f(\bq)$. 
The size of a query $\bq$, denoted $\norm{\bq}$, is the total number of tokens across all features, $\norm{\bq} = \sum_{f \in \bF} |f(\bq)|$. 
The size of a workload $\norm{\bS}$ is the sum of sizes of all queries in the workload.
For a workload $\bS$, we define $f(\bS) = \boldsymbol{\Large \biguplus}_{q \in \bS} f(\bq)$.
Finally, for a token $t \in \domain(\bS,f)$, its {\em frequency}, denoted $\mult{\bS}{t}{f}$ is the number of times the token appears in the multiset $f(\bS)$.

\subsection{Encoding Queries}

\label{subsec:encoding}

For the purpose of this paper, we consider a limited set of features that can be derived from a typical database system log entry, i.e.,~the SQL query text and its execution statistics. The features used throughout this paper are:

\vspace{0.5em}
\introparagraph{Categorical features} These are features derived from parsing the query statement, we choose:
\begin{inparaenum}[(1)]
    \item Function calls in the \texttt{\color{blue}SELECT} clause (such as \texttt{\color{blue}AVG,MAX} or some stored procedure),
    \item tables in the \texttt{\color{blue}FROM} clause of the (sub-)query,
    \item columns in the \texttt{\color{blue}GROUP BY} clause, and 
    \item columns in the \texttt{\color{blue}ORDER BY} clause.
\end{inparaenum}

\vspace{0.5em}
\introparagraph{Numeric features} These are features that describe the performance of a query, we choose:
\begin{inparaenum}[(1)]
    \item The total execution time of the query,
    \item planning time of the query,
    \item total size of the input to the query,
    \item total output rows of the query,
    \item CPU time spent executing the query, and 
    \item the number of joins as parsed from the query.
\end{inparaenum}
It is likely that for some numeric features such as execution time, no two queries have identical values.
To sparsify numeric features, we normalize the values such that they are in $[0, 1]$ by leveraging the largest and smallest values in the active domain.
Using this methodology allows us to reason about their discrete distribution. We transform the scaled numeric values into a histogram by assigning a bucket id $b_i \in \{ 0, \dots, H\}$ to each numeric value $v_i \in [0, 1]$ such that $b_i = \lfloor v_i \cdot H \rfloor$.
Observe that this transformation also changes the active domain to $\{ 0, \dots, H\}$ for the respective numeric features.

\vspace{0.5em}
\introparagraph{Example}
	Consider the workload $\bW = \{Q_1, Q_2, Q_3 \}$:
	\begin{center}
		\hspace*{-0.8em} $Q_1 = $ \texttt{\color{blue}SELECT} {\color{red} a}, \texttt{\color{blue}AVG}({\color{red} b}), \texttt{\color{blue}MAX}({\color{red} c}), \texttt{\color{blue}MAX}({\color{red} d}) \texttt{\color{blue}FROM} { \color{red} ${ {T_1}}$} \texttt{\color{blue}GROUP BY} {\color{red} a}
		
		\hspace*{-2.5em} $Q_2 = $ \texttt{\color{blue}SELECT} \texttt{\color{blue}COUNT}(*) \texttt{\color{blue}FROM} { \color{red} ${T_1}$}, { \color{red} ${T_2}$} \texttt{\color{blue}WHERE} { \color{red} ${T_1.a} = {T_2.a}$}
		
		\hspace*{-1em}
		$Q_3 = $\texttt{ \color{blue}SELECT} * \texttt{\color{blue}FROM} { \color{red}${T_1}$}, {\color{red}$T_2$}, {\color{red}${ {T_3}}$} \texttt{ \color{blue}ORDER BY} { \color{red} ${T_1.a}, {T_2.b},  {T_3.c}$}
	\end{center}
	
\noindent The corresponding domains for the categorical features are:
\begin{align*}
\domain(\bW, \mathit{f^c_1 = \mathsf{function\_call}}) & = \{ \texttt{\color{blue}AVG}, \texttt{\color{blue}MAX}, \texttt{\color{blue}COUNT}\} \\
\domain(\bW, \mathit{f^c_2 = \mathsf{table\_reference}}) & = \{ {\color{red}{T_1, T_2, T_3}} \} \\ 
%\domain(\bW, \mathit{f^c_3 = num\_joins}) & = \{0, 1, 2\} \\
\domain(\bW, \mathit{f^c_3 = \mathsf{group\_by}}) & = \{ {\color{red}{T_1.a}} \} \\
\domain(\bW, \mathit{f^c_4 = \mathsf{order\_by}}) & = \{{\color{red}T_1.a}, {\color{red}T_2.b}, {\color{red}T_3.c}\} 
\end{align*}

\noindent Table~\ref{syn:cat} shows the feature values per query. 
Note that for query $Q_1$, the token \texttt{\color{blue}MAX} appears twice, since it occurs two times in the selection clause. 

\begin{table}[t!]
	\centering
	\scalebox{0.9}{
	\begin{tabular}{c|c|c|c|c}
		\toprule[0.1em]
		& $f^c_1$ & $f^c_2$ & $f^c_3$ & $f^c_4$  \\ \midrule[0.1em]
		$Q_1$ & $\texttt{\color{blue}AVG}$,  $\texttt{\color{blue}MAX}$, $\texttt{\color{blue}MAX}$ & $\texttt{\color{red}$T_1$}$ &  $\texttt{\color{red}$T_1.a$}$ &   \\ \midrule
		$Q_2$ & $\texttt{\color{blue}COUNT}$ & $\texttt{\color{red}$T_1$}$, $\texttt{\color{red}$T_2$}$ & &  \\ \midrule
		$Q_3$ &  & 
		$\texttt{\color{red}$T_1$}$, $\texttt{\color{red}$T_2$}$, $\texttt{\color{red}$T_3$}$
		& &  
		\texttt{\color{red}$T_1.a$},
		 $\texttt{\color{red}$T_2.b$}$,
		 $\texttt{\color{red}$T_3.c$}$ \\				
		\bottomrule
	\end{tabular}%
	}
	\caption{Categorical features.} \label{syn:cat}
\end{table}

\noindent Additionally, we observe profile statistics as shown in Table~\ref{syn:fv}.
The numeric features $f^n_1, \dots f^n_6$ correspond to the total execution time, planning time, size of the input, total output rows, total CPU time, and the number of joins. 
For all numeric features, we normalize the values into a histogram with $H=10$, the resulting bucket assignment is shown in the lower part of Table~\ref{syn:fv}.

\begin{table}[t!]
	\centering
	\scalebox{0.9}{
	\begin{tabular}{c|c|c|c|c|c|c}
		\toprule[0.1em]
		& $f^n_1$ & $f^n_2$ & $f^n_3$ & $f^n_4$ & $f^n_5$ & $f^n_6$  \\ \midrule[0.1em]
		$Q_1$ & $5$ms & $4$ms& $5$MB & $100$ & $2$ms & 0  \\ \midrule
		$Q_2$ & $10$ms & $2$ms & $10$MB & $1000$ & $3$ms & 1  \\		\midrule
		$Q_3$ & $8$ms & $5$ms & $20$MB & $500$ & $4$ms & 2 \\	
		\bottomrule
		\multicolumn{6}{l}{} \\
		\toprule[0.1em]
		$Q_1$ & $0$ & $6$ & $0$ & $0$ & $0$ & 0 \\ \midrule
		$Q_2$ & $10$ & $0$ & $6$ & $10$ & $5$ & 5 \\		\midrule
		$Q_3$ & $6$ & $10$ & $10$ & $4$ & $10$ & 10 \\	
		\bottomrule
	\end{tabular}%
	}
	\caption{Extracted numeric feature values.} \label{syn:fv}
\end{table}

\noindent The sizes of the queries are as follows: $\norm{Q_1} = 11$, i.e.,~six numeric feature tokens, three tokens in $f_1^c$, and one in $f_2^c$ and $f_3^c$ respectively, $\norm{Q_2} = 9$ and $\norm{Q_3} = 12$. The size of the workload is $\norm{\bW} = 32$.

\rtwo{
\noindent \introparagraph{Extension to other features} Our techniques are not limited to presented features, but we simply choose these for demonstration purposes. 
Unlike the popular summarization scheme introduced by Aligon et al.~\cite{aligon2014similarity}, we do not restrict the features to relation names and columns in the \texttt{{\color{blue} WHERE}}, \texttt{{\color{blue}SELECT}}, or \texttt{\color{blue}FROM} clause. 
Thus, in principle, any feature can be used in our framework. 
However, in practice, the choice of features is limited by the hardware and available resources.
For example, if a \texttt{GPU} is used for some of the queries, it would be useful to add features such as \texttt{average\_memory\_bandwidth\_used}. 
%For queries that do not use the \texttt{GPU}, the absence of the feature $f$ can be encoded by a special value $f^\phi$. 
One could also create a higher-order feature derived from two different features $f = f_1 \times f_2$ that captures the co-occurrence of $\langle t_1, t_2 \rangle$ where $t_1$ and $t_2$ are tokens of $f_1$ and $f_2$. 
Capturing cross-feature information may lead to improved summaries, but defining such features requires a principled approach to feature engineering (see Section~\ref{sec:limitations} for more details). 
It is also possible to encode features such as query plan fragments, indexes used during query evaluation, and physical execution operators using standard techniques of 1-hot encoding and transforming a query plan into a tree of vectors (see Sec 3.2 in~\cite{marcus2019neo}) for more details.}

\subsection{Metrics}
\label{sec:ps}

We next formalize the definitions of the coverage and representativity metrics that are subsequently used throughout this paper.

\smallskip
\introparagraph{Coverage} 
Given a feature $f \in \bF$, the coverage factor $\alpha_f$ is defined as the fraction of tokens covered by the compressed workload for feature $f$. 
To generalize this to a metric across all features, we can either compute the minimum or average $\alpha_f$. 
Formally:

\begin{definition}[Coverage]
Let $\bS$ be a summary of the workload $\bW$. The {\em coverage factor} for a feature $f \in \bF$ is defined as
$\alpha_f = \frac{\vert \domain(\bS, f) \vert}{ \vert \domain(\bW,f) \vert}$. The {\em minimum coverage factor} and 
{\em average coverage factor} are respectively defined as:
\begin{align*}
\alpha_{\min} = \min_{f \in \bF} \alpha_f, \quad \quad \alpha_{\text{avg}} = \sum_{f \in \bF} \alpha_f /\vert \bF \vert
\end{align*}
\end{definition}

Observe that both the minimum and average coverage values are always in $[0,1]$. A score of $1$ means that the coverage is perfect.

\smallskip
\introparagraph{Representativity} 
A representative summary of the workload must capture the structural properties of the original workload. 
Specifically, workload $\bW$ induces a discrete distribution $p_{\bW}(\cdot)$ over the tokens present in the features of the queries in the workload. 
In particular, for any token $t$, 
\begin{align*}
p_{\bW}(t) = \frac{\mult{\bW}{t}{f}} {\sum_{f} \sum_{v \in \domain(\bW, f)}  \mult{\bW}{v}{f}} 
\end{align*}
In other words, $p_{\bW}(t)$ denotes the probability of selecting token $t$ if we choose a token from $\bW$ uniformly at random. 
The summary $\bS$ will induce a distribution $p_{\bS}(\cdot)$; the representativity metric then measures the distance between $p_{\bS}$ and $p_{\bW}$. 
In general, we wish $p_{\bS}$ to be as close to some target distribution $\td(\cdot)$.
Note that the target distribution can be different from $p_{\bW}$ if wanted.

\begin{definition}[Representativity]
Given a target token distribution $\td$, the {\em representativity} w.r.t. $\td$ is defined using the following two metrics:
	\begin{align*}
		\rho_{1}(\td) &= 1 - \frac{1}{2}  \sum_f {\sum_{t \in \domain(\bW, f)} \vert p_{\bS}(t) - \td(t) \vert } \\
		\rho_{\infty}(d) &= 1 - \max_f \max_{t \in \domain(\bW, f)} \vert p_{\bS}(t) - \td(t) \vert 
	\end{align*}
\end{definition}

The $\rho_1$ metric essentially measures the total variation distance between the two distributions and is a popular distance metric for graph visualizations~\cite{macke2018adaptive}. $\rho_{\infty}$ metric captures the largest deviation in the distribution across all features.
If the representativity score is $1$, we say that the compressed workload is perfectly representative. Note that there are other possible definitions of representativity. We refer the reader to~\autoref{sec:a} for more discussion.

\smallskip
\introparagraph{User-specific modifications} 
If users have specific domain knowledge, they may want to use a $(i)$ weighted version of computing these metrics and/or, $(ii)$ target distribution for representativity.
Both of these modifications are supported in our framework.
For the former, we require that each feature is assigned a weight $w_f$ such that $\sum{w_f} = 1$.
For the latter, we define a general version of the metrics w.r.t. an arbitrary target distribution $\td$; the case where $\td$ is the same as the input distribution becomes a special case. \rtwo{This functionality is useful in applications such as test workload generation. Developers  frequently use queries to test their code while developing the functionality in RDBMS. However, instead of choosing from a set of predefined queries, it is more desirable to choose the test workload from a set of production queries, which increases more confidence in the testing of the functionality. Due to  lack of space, we refer the reader to the appendix for more details}.

\vspace{0.5em}
\introparagraph{Example}  
Consider the setup from our running example and let $\bS = \{Q_1, Q_3\}$. 
The normalizing factor for $\bW$ and $\bS$ is $14 + 18 = 32$ ($14$ categorical tokens and $18$ numeric feature tokens) and $11 + 12 = 23$ respectively. 
\Cref{tab:fcf} and \Cref{tab:trf} show the $p_{\bW}$ and $p_{\bS}$ distributions (target $\td$ is set to be the input distribution) for the function call and table reference features.

\smallskip
The reader can verify that $\rho_{1} = 1 - \frac{1}{2} \big\{ 18 (\frac{1}{23} - \frac{1}{32}) + \frac{14}{32} - \frac{5}{23}\big\} \approxeq 0.779$ and $\rho_{\infty} = 1 - \frac{1}{32} \approxeq 0.96$.

\smallskip
Similarly, Table~\ref{syn:coverage} shows $\alpha_f$ for each feature. For example, only two tokens of the function call feature are covered by $\bS$ and $\langle${\color{blue} \texttt{COUNT}}$\rangle$ is missed since $Q_2 \not \in \bS$. All numeric features have $\alpha_f = \frac{2}{3}$. Thus, the minimum coverage factor is $\alpha_{\min} = \frac{2}{3}$ and the average coverage factor is $\alpha_{\mathit{avg}} = \frac{23}{30}$.

\begin{table}
\centering
\scalebox{0.95}{
    \begin{subtable}{0.48\columnwidth}
    \centering
        \begin{tabular}{c|c|c}
			\toprule[0.1em]
			token & $p_\bW(t)$ & $p_\bS(t)$ \\ 
			\midrule[0.1em]
			\texttt{\color{blue}AVG} & $0.031$ & $0.043$ \\ 
			\midrule
			\texttt{\color{blue}MAX} & $0.062$ & $0.086$ \\ 
			\midrule
			\texttt{\color{blue}COUNT} & $0.031$ & $0$ \\			
			\bottomrule
		\end{tabular}
        \caption{Function call feature.}
        \label{tab:fcf}
    \end{subtable}}
    \hfill
    \scalebox{0.95}{
    \begin{subtable}{0.48\columnwidth}
    \centering
        \begin{tabular}{c|c|c}
		\toprule[0.1em]
		token & $p_\bW(t)$ & $p_\bS(t)$ \\ 
		\midrule[0.1em]
		\texttt{\color{red}$T_1$} & $0.093$ & $0.086$ \\ 
		\midrule	
		\texttt{\color{red}$T_2$} & $0.062$ & $0.043$ \\ 
		\midrule				
		\texttt{\color{red}$T_3$} & $0.031$ & $0.043$ \\
		\bottomrule
	\end{tabular}
        \caption{Table reference feature.}
        \label{tab:trf}
    \end{subtable}}
    %\vspace{1em}
    \scalebox{0.95}{
    \begin{subtable}{\columnwidth}
    \centering
        \begin{tabular}{ c|l|l|l|l|l|l|l|l|l|l}
		\toprule[0.1em]
		cov & $f^c_1$ & $f^c_2$ & $f^c_3$ & $f^c_4$ & $f^n_1$ & $f^n_2$ & $f^n_3$ & $f^n_4$ & $f^n_5$ & $f^n_6$  \\ \midrule[0.1em]
		$\alpha_f$ & $\frac{2}{3}$ & $1$ & $1$ & $1$ & $\frac{2}{3}$ & $\frac{2}{3}$ & $\frac{2}{3}$ & $\frac{2}{3}$ & $\frac{2}{3}$ & $\frac{2}{3}$ \\				
		\bottomrule
	\end{tabular}%
	\caption{Feature coverage for workload $\bS$.} \label{syn:coverage}
    \end{subtable}}
    \caption{Representativity and coverage computation.}
    %\vspace{-1em}
\end{table}

\subsection{Problem Statement} \label{sec:statement}
Given an input workload $\bW = \{ q_1, \dots, q_n\}$ where each query $\bq$ in $\bW$ is associated with a non-negative cost $c(\bq)$ such as the \rtwo{size} of the summary workload.

Assuming a target distribution $\td(\cdot)$ over these tokens, a budget constraint $B \geq 0$, and a parameter $\beta \in [0,1]$, our goal is to construct a {\em summary workload}  $\bS \subseteq \bW$ such that:
\begin{compactitem}
\item the cost of the summary workload is less than the budget, $\sum_{\bq \in \bS} c(\bq) \leq B$; and 
\item the quantity $\beta \cdot \alpha+ (1- \beta) \cdot \rho(d)$ is maximized, where $\alpha \in \{\alpha_{min}, \alpha_{avg}\}$ and $\rho \in \{ \rho_1, \rho_\infty\}$.
\end{compactitem}

Here, the user-specified parameter $\beta$ controls the trade-off between the coverage and representativity metrics. 
If $\beta=0$, we optimize for representativity only, and if $\beta =1$, we optimize for coverage. 
The compression ratio, $\eta = 1 - c(\bS)/c(\bW)$, is the fraction of queries that have been pruned.
Observe that the larger the value of $\eta$, the smaller the compressed workload.

\section{Design Considerations}
\label{subsec:desiderata}
Several desiderata are important to consider when compressing a workload, and these form a rich design space.
We now describe these desiderata and their corresponding trade-offs.

\smallskip
 \introparagraph{High Coverage} 
 High coverage is desirable to ensure that long-tail feature values are part of the compressed workload. Ideally, we would like to maximize the coverage subject to certain budget constraints. Maximizing coverage is an NP-hard problem~\cite{nemhauser1978analysis, feige1998threshold} that can be efficiently approximated~\cite{nemhauser1978analysis}. 

\smallskip
 \introparagraph{High Representativity} 
 High representativity implies that the compressed workload must faithfully reproduce the target distribution, which can be either derived from the input workload's feature distribution or specified by the user.
 This is a key requirement for successful workload compression when used, for example, in the context of performance analysis in a database system.

\smallskip
 \introparagraph{Customizability} Representativity and coverage are competing objectives, which makes the task of maximizing both metrics simultaneously hard. 
 For instance, simple random sampling achieves high representativity but may miss long-tail queries. 
 On the other hand, set cover algorithms maximize coverage but will not pick queries whose features have already been covered. 
 It is therefore desirable that the user can control this trade-off smoothly, which we realize through the parameter $\beta$.

\smallskip
 \introparagraph{User-Specific Constraints} Users may want to specify constraints on some property of the compressed workload. 
 For instance, the user may want to limit the size, total execution time, or the representativity target distribution of the compressed workload. 
 The framework should be flexible enough to allow users to specify these constraints on-the-fly. 
 For simplicity, we restrict ourselves to two types of constraints:
$(i)$ specifying the desired feature distribution of the summary workload; and 
$(ii)$ non-negative modular constraints (i.e.,~a knapsack constraint) of the form $\sum_{\bq \in \bS} c(\bq) \leq B$, where $c(\cdot)$ can be any cost function.

\smallskip
 \introparagraph{Scalability} Efficient computation of the compressed workload is a key requirement for any framework to be deployed in practice. 
 Ideally, the compression algorithm must compute the compressed workload fast and scale effectively to large input workloads. 
 This will also allow the user to find the correct parameter settings for fine-tuning the representativity and coverage trade-off dynamically. 

\smallskip
 \introparagraph{Incremental Computation} Consider a user who wants to analyze how the workload is changing over time with respect to a set of features. 
For this case, we want to avoid computing the compressed workload from scratch every day; instead, it would be better to create a summary for each day, and then merge them. 
 In other words, we would like to construct {\em mergeable} compressed workloads.

	\section{Hardness Results}
\label{sec:hardness}

In this section, we show that our problem is computationally hard for any parameter $\beta \in [0,1]$, even for the simple case where the cost function is the same constant for every query, \ie,~$c(\bq)=1$. 
We note that~\cite{feige1998threshold} already shows that the problem is NP-complete when $\beta=1$ (\ie,~we want to maximize only coverage) for both coverage metrics $\alpha_{min}$ and $\alpha_{avg}$. The next theorem shows that the NP-hardness result extends for any choice of the parameter $\beta$.

\begin{theorem} \label{thm:hardness}
Let $\alpha \in \{\alpha_{min}, \alpha_{avg}\}$, $\rho \in \{ \rho_1, \rho_\infty\}$, $\beta \in [0,1]$, and $\td(\cdot)$ be a target distribution. 
Then, the problem of finding a summary $\bS \subseteq \bW$ such that $\vert \bS \vert \leq B$ and the quantity $\beta \cdot \alpha+ (1- \beta) \cdot \rho(d)$ is maximized is NP-complete. 
In particular, the problem remains NP-complete when $\td$ is the input distribution $p_{\bW}$, and there exists only one multi-valued feature.
\end{theorem}
%\begin{proof}
%	We outline the reduction from a variant of the XSAT (exact SAT) problem. In XSAT, we are given a CNF formula such that every clause has size exactly $k \geq 3$, and each variable is present in exactly
%	$l \geq 3$ clauses. Further, each clause has only positive variables (so no negated variables). The goal is to find a satisfying assignment such that exactly one variable in each clause is true and all other variables in that clause are false. This problem is known to be NP-hard (\cite{porschen2014xsat}, Theorem 29). Let the input instance for XSAT be a formula $\phi$ with $m$ clauses. We construct a workload $\bW$, by introducing a query $\bq_x$ for each
%	variable $x$ that occurs in $\phi$. We consider a single multivalued feature $f$, where each token $t_c$ corresponds to a clause $c$.  For a query $\bq_x$, $f(\bq_x)$ contains exactly the clauses (tokens) that contain variable $x$. Since each clause contains exactly $k$ variables, we have that for any token $t$, $p_{\bW}(t) = k/(k \cdot m) = 1/m$. We set the target distribution to be
%	$q(t) = p_\bW(t)$, and choose the budget to be $B = m$. It can be shown that XSAT has a solution if and only if the quantity $\beta \alpha_{min} + (1-\beta) \rho_1$ has maximum $1$.
%\end{proof}

Since the problem of maximizing the objective function with respect to a cost constraint is NP-hard, an approximation algorithm with theoretical guarantees is required to solve the problem. 
Indeed, if we restrict to optimize only for coverage (so $\beta=1$) and a single feature, there exists a greedy algorithm (by~\cite{feige1998threshold}) that achieves an $(1-1/e)$-approximation ratio. 
%However, it is unclear how this approximation algorithm extends to multiple features. 
Next, we show that the problem is APX-hard for any choice of parameter $\beta$.
 \begin{theorem} \label{thm:hardness:apx}
	Let $\alpha \in \{\alpha_{min}, \alpha_{avg}\}$, $\rho \in \{ \rho_1, \rho_\infty\}$, $\beta \in [0,1]$, and $\td(\cdot)$ a target distribution.  Then,
	the problem of finding a summary $\bS \subseteq \bW$ such that $\vert \bS \vert \leq B$ and the quantity 
	$\beta \cdot \alpha+ (1- \beta) \cdot \rho(d)$ is maximized is APX-hard.
\end{theorem}
%\begin{proof}
%	The problem is known to be APX-hard for $\beta = 1$, since it reduces to maximum coverage problem~\cite{feige1998threshold}. 
%	We show the hardness of maximizing the objective $\beta \alpha_{min} + (1-\beta) \rho_1$ for any $\beta \in [0,1)$ using a gap preserving (PTAS) reduction from the restricted \textsf{MAX-3DM} problem~\cite{petrank1994hardness}.
%	
%	\smallskip
%	\noindent \textbf{Input.}  Three disjoint sets $A = \{a_1, \dots, a_p\}, B = \{b_1, \dots, b_p\}$ and $C = \{c_1, \dots, c_p\}$, together with a subset of triples $T \subseteq A \times B \times C$. Any element from $A,B,C$ occurs in at most 3 triples in $T$. This implies $p \leq |T| \leq 3\cdot p$.
%	
%	\smallskip
%	\noindent \textbf{Goal.} Find a subset $T' \subseteq T$ of maximum cardinality such that no two triples of $T'$ agree in any coordinate.
%	%
%	\smallskip
%	
%	Petrank~\cite{petrank1994hardness} has shown that \textsf{Max-3DM} is APX-hard even if one only allows instances where the optimal solution consists of $p = |A| = |B| = |C|$ triples; it can then be shown that there exists a PTAS reduction to prove our claim. 
%	
%%	To achieve the PTAS reduction, we map an instance $I$ of \textsf{Max-3DM} to an instance $I'$ of the summarization problem as follows. We introduce a single multi-valued feature with domain $A \cup B \cup C$. For each triple, we create a query with feature set the values that appear in that triple. We pick the target distribution to be $\td(t) = 1/3p$ for every token $t \in A \cup B \cup C$, and use the constraint $|S| = p$.
%\end{proof}
% 
The problem gets even more complex if representativity is taken into account.
As the next lemma shows, neither $\rho_1$ or $\rho_\infty$ metrics satisfy desirable properties from an optimization perspective.

\begin{lemma}
The $\rho_{1}$ and $\rho_\infty$ metrics are not monotone or submodular.
\end{lemma}

%The non-monotonicity of representativity is not just a technical issue. 
%Specifically, it is possible that given the constraint that the summary workload is of size at most $k$, the optimal summary workload has a size smaller than $k$. 

%\begin{example}\label{ex:non-monotone}
%Suppose that a workload consists of $88$ queries with $22$ unique query templates. 
%Intuitively, a summary workload should contain each template once. 
%If the constraint is to generate a summary workload of size $33$ (at most), then the best summary workload is of size $22$ (and $\rho_1 = 1$), where each query template is present exactly once. 
%However, if we choose a workload of size $33$, then $\rho_{1}$ will be closer to $0.5$ since half of the feature values will occur roughly twice as many times than they should.
%\end{example}
	\section{Problem Solution}
\label{sec:main}

There are two often-applied methods to solve the summarization problem: {\em clustering} and {\em random sampling}.
While clustering-based methods (such as k-medoids\footnote{K-medoids is an iterative greedy algorithm that chooses $k$ cluster centers, assigns all points to the closest center  and iteratively refines the points in each cluster.} and hierarchical clustering\footnote{Hierarchical clustering is a top-down approach where all points start in a single cluster and the algorithm recursively splits the points into $k$ disjoint clusters.}) identify the patterns in the workload, they suffer from the following drawbacks: $(i)$ $O(n^2)$ time complexity, $(ii)$ sensitivity to the distance function and $(iii)$ the number of clusters $k$ is required as an input. 
The best value for $k$ is not known a priori. 
To address this problem, one commonly used idea is to run the clustering algorithm several times, where the cluster size is doubled in every iteration. 
However, this may be far from optimal because of the sensitivity of the metrics to the size constraint. To address the drawbacks of clustering and random sampling, we present a new approach to summarization.
We define a new objective function (Section~\ref{subsec:obj}) that can be parametrized to control the trade-off between coverage and representativity followed by efficient algorithm (Section~\ref{sec:greedy}). 

\subsection{A New Objective Function} \label{subsec:obj}

Instead of using the initial objective of the summarization problem, we replace it with the following objective, where 
$\gamma \in (0,1]$ is a smoothing parameter that controls the trade-off between representativity and coverage:
\begin{align} \label{eq:objective}
G(\bS, \gamma) = \sum_f \sum_{t \in \domain(\bW, f)} \td(t) \cdot \log  \left( \dfrac{\mult{\bS}{t}{f} +\gamma} { \gamma} \right)
\end{align}

Before we explain the intuition of choosing this objective, we show that it satisfies several properties that make it amenable to optimization. 
In particular, $G(\bS, \gamma)$ is a non-negative, monotone and submodular set function. 

\begin{proposition}
For any value $\gamma \in (0,1]$,  the set function $G(\bS, \gamma)$ is non-negative, monotone and submodular.
\end{proposition}

%\begin{proof}
%First, observe that for $\bS = \emptyset$, we have $G(\bS, \gamma) = 0$.
%Since $(i)$ the $\log$ function is submodular and monotone (increasing), and $(ii)$ the class of monotone submodular functions is closed under linear combinations, the resulting function is also monotone submodular. 
%Since $G$ is monotone and has value 0
%for the empty set, it is also non-negative.
%\end{proof}	

\smallskip
\introparagraph{Analysis of the Objective}
To understand the intuition behind the choice of the objective function, we first discuss how  $G(\bS, \gamma)$ behaves for very small values of $\gamma$.

\begin{lemma} \label{lem:sum:comparison}
Let $\bS_1, \bS_2$ be two summaries of a workload $\bW$ such that
$ \bigcup_f \domain(\bS_1, f)  \subsetneq  \bigcup_{f} \domain(\bS_2, f)$.
Then, for $\gamma \rightarrow 0$ we have $G(\bS_1, \gamma) < G(\bS_2, \gamma)$.
\end{lemma}

Lemma~\ref{lem:sum:comparison} tells us that when the parameter $\gamma$ tends to zero, a summary that covers strictly more tokens will always have a better value for the objective $G$, independent of the representativity of each summary. 
This implies that if there exists a summary $\bS$ within the budget $B$ that covers all tokens of $\bW$, then an optimal solution for $G$ will always cover all tokens as well.
Now, let us consider a summary $\bS$ that achieves perfect coverage. We can then write:
\begin{align*}
& \lim_{\gamma \rightarrow 0} \{ G(\bS,\gamma) + \log \gamma \} = 
\sum_f \sum_{t \in \domain(\bW, f)} \td(t) \cdot \log \mult{\bS}{t}{f} \\
&= - \sum_f \sum_{t \in \domain(\bW, f)} \td(t) \cdot \log \dfrac{\td(t)}{p_{\bS}(t)} - H(\td) + \log \norm{\bS} \\
&= -KL(\td \lVert p_{\bS}) - H(\td) + \log \norm{\bS}
\end{align*}
where $KL(d \lVert p)$ is the Kullback-Leibler (KL) divergence, a metric that captures the difference of the two distributions:
\begin{align*}
	KL(\td \lVert p) = \sum_{x \in \Omega} \td(x) \ln \frac{\td(x)}{p(x)}
\end{align*} 

Thus, when $\gamma \rightarrow 0$, among all summaries with the same size and perfect coverage, the objective prefers the one that minimizes the KL divergence between the target distribution and the summary. 
We should note here that KL divergence is related to the total variation distance by the well-known Pinkser's inequality: $TV(\td, p) \leq \sqrt{\frac{1}{2} KL(\td \lVert p)}$. 
If the summaries do not have the same size, then the summary size will also influence the objective.

%The Kullback-Leibler (KL) divergence captures the difference of the two distributions accounting for the bulk of the distributions:
%	\begin{align*}
%	KL(p \lVert q) = \sum_{x \in \Omega} p(x) \ln \frac{p(x)}{q(x)}
%	\end{align*} 

As $\gamma$ increases from 0 to 1, the penalty for not covering a token decreases. 
Hence, an optimal solution will focus less on maximizing coverage and more on maximizing representativity. 
For larger values of $\gamma$, the objective function will choose the summary that minimizes the KL divergence between the target distribution $\td$ and the `smoothed' summary distribution where the probability of a token is proportional to $\mult{\bS}{t}{f} + \gamma $ instead of the frequency $\mult{\bS}{t}{f}$. 
Intuitively, one can think of the case of $\gamma =1$ as if each token already starts with a count of 1 as the summary is constructed.

\subsection{A Greedy Algorithm} \label{sec:greedy} 
We now present an algorithm that solves our optimization problem which can be formally stated as follows:
\begin{align*}
 \textbf{maximize} & \quad G(\bS,\gamma) \\
  \textbf{subject to } & \quad \sum_{\bq \in \bS} c(\bq) \leq B, \quad \bS \subseteq \bW
\end{align*}

We solve the above optimization problem greedily in Algorithm~\ref{algo:baseline}.
%We first present a simply baseline greedy strategy , and then show how to optimize its execution (Algorithm~\ref{algo:final}). 
The algorithm starts with an empty summary $\bS_0 = \emptyset$. 
At the $i$-th iteration of the main loop, it adds the query from the workload that maximizes the {\em normalized marginal gain} $\Delta(\bq \mid \bS_{i-1})$ to the current summary $\bS_{i-1}$.
The normalized marginal gain is defined as 
$$ \Delta(\bq \mid \bS) =  \frac{G(\bS \cup \{ \bq \}, \gamma) - G(\bS, \gamma)}{c(\bq)}.$$
In other words, the algorithm greedily chooses the query with the best gain per unit of cost.
To increase the efficiency of the algorithm, we apply a common optimization~\cite{leskovec2007cost} which skips the computation of the normalized gain $ \Delta(\bq \mid \bS_{i-1})$ of a query $\bq$ at round $i-1$ if we know that the gain can not be larger than the query with highest gain so far (line~\ref{line:skip}). 
This optimization works because submodularity tells us that the gain can only decrease as the size of the summary grows (hence, values of $\Delta(\bq \mid \bS_{k})$ for $k < i-1$ are an upper bound to the gain).
Experiments in the appendix show that  this lazy strategy can substantially speed up execution.
Algorithm~\ref{algo:baseline} considers additionally the best single element solution, and chooses the best of the two (line~\ref{line:bestof2}). 
Since $G$ is a monotone, non-negative and submodular function, it can be shown that Algorithm~\ref{algo:baseline} achieves an $1/2(1 - 1/e)$ approximation guarantee~\cite{leskovec2007cost, krause2005note}.

\begin{algorithm}[htp]
	\DontPrintSemicolon
	\LinesNumbered
	\SetNoFillComment
	\SetKwInOut{Input}{\textsc{input}}\SetKwInOut{Output}{\textsc{output}}
	\Input{input workload $\bW$, cost function $c$, budget $B$, parameter $\gamma \in (0,1]$}
	\Output{summary workload $\bS$}
	$\bS \leftarrow \emptyset$ \;
	$\forall \bq \in \bW : \Delta(\bq) \leftarrow 0$ \;
	\While{$\bW \neq \emptyset$}{
		$\Delta^* \leftarrow -1$ \;
		\ForEach{$\bq \in \bW$}{
		\If{$\Delta(\bq) > \Delta^*$ \label{line:skip} }{
		$\Delta(\bq) \leftarrow  \frac{G(\bS \cup \{\bq\}, \gamma) - G(\bS,\gamma)}{c(\bq)}$ \;}
		\If{$\Delta(\bq) > \Delta^*$}{
			$\Delta^*\leftarrow \Delta(\bq)$ \; 
			$\bq^* \leftarrow \bq$ \;
			}
		}
		\If{$c(\bS) + c(\bq^*) \leq B$}{$\bS \leftarrow \bS \cup \{\bq^*\}$ \;}
		$\bW \leftarrow \bW \setminus \{ \bq^* \}$ \;
	}
	$\bS' \leftarrow \mathit{argmax}_{\bq \in \bW}\{G(\{\bq\},\gamma) \mid c(\bq) \leq B \}$ \;
	\KwRet{$\mathit{argmax}_{\bS,\bS'} \{G(\bS,\gamma), G(\bS',\gamma) \}$} \label{line:bestof2}
	\caption{Greedy algorithm}
	\label{algo:baseline}
\end{algorithm}

\smallskip
\introparagraph{Runtime Analysis} The runtime cost of the algorithm is dominated by the cost of the main loop. During each iteration, the algorithm needs to compute the normalized marginal gain for each of the $n$ queries (in the worst case). Since the feature vector is sparse, each iteration of the main loop can be performed in $O(n)$ time. The number of iterations can be as large as $n$, resulting in a worst-case runtime of $O(n^2)$. 
However, we can obtain better bounds depending on the budget constraint $B$ and the cost function $c(\cdot)$. 
For example, if $c(\cdot)$ is the unit cost function, then the number of iterations can be at most $B$, and the runtime becomes $(n \cdot B)$. 
In general, if $c_{min}$ is the smallest possible cost of the query, then the number of iterations is upper bounded by $B/c_{min}$. If we want to optimize for our original score function $\alpha \beta + (1-\beta) \rho(q)$, observe that the summary we obtain at the end of the algorithm may not be the best one. We can slightly modify Algorithm~\ref{algo:baseline} by recording the best score and the corresponding set that achieves it at every iteration without any impact on the total running time. 

%\smallskip
%\introparagraph{Optimizing for the Original Objective} If we want to optimize for our original score function $\alpha \beta + (1-\beta) \rho(q)$, observe that the summary we obtain at the end of the algorithm may not be the best one (\ie, the objective was maximized in an earlier iteration). 
%To take this into account, we can slightly modify Algorithm~\ref{algo:baseline} by recording the best score and the corresponding set that achieves it at every iteration. 
%Note that the operation of computing the score at every iteration is an $O(n)$ time operation in the worst case. 
%From our experience, the worst case is rarely triggered as the numeric features have a bounded active domain of size $H$ and categorical features are also constrained by SQL syntax (for example a small number of aggregate and group by operators). 

\section{Parallelization and Incremental Computation} \label{sec:parallel}
Our algorithm is inherently parallelizable and well suited for incremental computation.

\smallskip
\introparagraph{Parallelization} Consider the problem where our cost function is $c(\bq)=1$, and the budget is $B$.
In order to parallelize Algorithm~\ref{algo:baseline}, we partition the input workload into $B$ machines and run the algorithm on each machine. 
Each run results in $B$ different summary workloads, one from each machine. 
We then merge each machine's summary workload, which will act as the new input workload to generate the final summary. 
Observe that the first step of generating $B$ different summary workloads takes $O(\frac{n}{B} \cdot B) = O(n)$ time in parallel, while the second step of merging requires time $O(B^2 \cdot B) = O(B^3)$. 
Hence, we obtain a faster runtime if $B^3 \leq n \cdot B$, i.e $B \leq \sqrt{n}$.
For values of $B \geq \sqrt{n}$, there exist algorithms that allow for parallelization with slightly worse approximation guarantees.
We refer the reader to~\cite{badanidiyuru2014streaming, mirzasoleiman2016fast, mirzasoleiman2015lazier, mirzasoleiman2013distributed} for a more detailed discussion on parallelizing submodular maximization problems. 
%The above parallelization strategy works for alternative cost functions $c$ and budget $B$ as well, but without any formal guarantees on speedup.

\smallskip \introparagraph{Incremental Computation} 
Supporting incremental computation of a summary is critical in the case where the workload that needs to be summarized is not provided at once, but instead constantly grows. Since we want to perform summarization {\em across} multiple workloads over time, we need to normalize for numeric features in a consistent way, \ie,~the maximum and minimum values used to normalize need to be fixed a priori.  In order to do so, we fix the largest and smallest value for all numeric features explicitly. 
Although this assumption may feel restrictive, in our experience, setting the maximum and minimum values for a feature by looking at historical workloads works very well. 
For instance, it is safe to assume that the number of joins in a query will be smaller than $1000$ in ad-hoc workloads. 
Suppose now that we have computed a summary $\bS$ of $\bW$, and a new set of queries $\bW'$ is added (with the same feature set) to the current workload. 
Instead of computing directly the summary of $\bW \boldsymbol{\Large \biguplus} \bW'$, we can compute a summary $\bS'$ of $\bW'$, and merge the two summaries to obtain 
$\bS \boldsymbol{\Large \biguplus} \bS'$. 
The next two lemmas describe how the merged summary behaves:

\begin{lemma} \label{lem:merge}  
Let $\bS$ be a summary for $\bW$ with $\alpha_{min} = \alpha$ and $\rho_\infty(p_{\bW}) = \rho$.
Also, let $\bS'$ be a summary for $\bW'$ with $\alpha_{min} = \alpha'$ and $\rho_\infty(p_{\bW'}) = \rho'$.
Then, $\bS \boldsymbol{\Large \biguplus} \bS'$ is a summary of $\bW \boldsymbol{\Large \biguplus} \bW'$ 
 that has $\alpha_{min} \geq  \min\{\alpha, \alpha'\} / 2$, $\rho_\infty \geq \min\{\rho, \rho'\}$ and cost $c(\bS) + c(\bS')$. 
\end{lemma}

\begin{lemma}
	Consider two summaries $\bS_1$ and $\bS_2$ for $\bW_1, \bW_2$ respectively) with identical budget $B$. Then, we can produce a summary $\bS$ that is a subset of $\bS_1 \boldsymbol{\Large \biguplus} \bS_2$, such that its cost is at most $B$ and its objective value is 
	%an algorithm that produces $\mathbf{S_1} \boldsymbol{\oplus} \mathbf{S_2}$ with budget $B$ in time $O(B \cdot \vert \mathbf{S_1} \boldsymbol{\oplus} \mathbf{S_2} \vert)$. Further, $G(\mathbf{S_1} \boldsymbol{\oplus} \mathbf{S_2})$ is 
	at most an $1/2(1-1/e)$ factor away from the optimal solution of $G$ for $\bW_1 \boldsymbol{\Large \biguplus} \bW_2$.
\end{lemma}

As we will see in the experimental evaluation, the worst-case bounds do not occur in practice and incremental merging of the summaries works well.

\rtwo{\section{End-To-End Framework} \label{sec:endtoend}
Benchmarking is an important problem to solve in a structured manner as it allows developers and users to reason about the performance of a system over time.
\texttt{DIAMetrics}~\cite{diametrics} is an end-to-end benchmarking system developed at Google for query engine-agnostic, repeatable benchmarking that is indicative of large-scale production performance.
In essence, it allows users to $(a)$ anonymize production data, $(b)$ move data between different file storage systems, $(c)$ execute preset workloads on specific systems automatically, and $(d)$ visualize the results of executed benchmarks.
\texttt{DIAMetrics} provides the context for which \framework was prototyped. 

One of the biggest barriers of entry to benchmarking a system is that teams are often unable to provide a concise benchmark that represents their production workload.
Although clustering and simple frequency-based analysis has been sufficient for some cases, in a majority of the cases it is infeasible to manually create accurate benchmarks. 
Workload compression provides a powerful means to generate a subset of production queries with formal guarantees.
In addition to its usefulness for benchmarking,~\framework can be deployed daily to build workload summaries which are used to monitor workload patterns over time. 
A shift in these patterns can entail several issues such as changing resource usage or execution regressions which need to be addressed in a timely manner.
A full detailed description of \texttt{DIAMetrics} is beyond the scope of this paper and we refer the interested reader to~\cite{diametrics} for more details and use cases.}

%\rthree{\texttt{DIAMetrics} also allows easy visualization of feature distributions of a given workload which aids the benchmark designers to visualize how close a compressed workload distribution is to the input distribution.} 

	% !TeX root = paper.tex
\section{Evaluation} \label{sec:eval}

%\begin{figure}[t]
%    \includegraphics[width=\columnwidth, scale=0.6]{figures/Architecture.pdf}
%    \caption{Summarization framework architecture}
%    \label{fig:arch}
%\end{figure}

In this section, we empirically evaluate the techniques discussed throughout this paper. More specifically, we
\begin{packed_item}
	\item validate that the summarization framework is useful for index tuning, materialized view recommendations, and test workload generation.
	\item evaluate the runtime of all algorithms for varying workloads and summary size constraints.
	\item compare the coverage and representativity metrics of our solution with k-medoids, hierarchical clustering, and random sampling both on real production workloads and standardized workloads and its scalability.
	\item evaluate the trade-off between representativity and coverage.
\end{packed_item}

\begin{figure*}[t]

	\begin{subfigure}{0.32\linewidth}
		\includegraphics[scale=0.35]{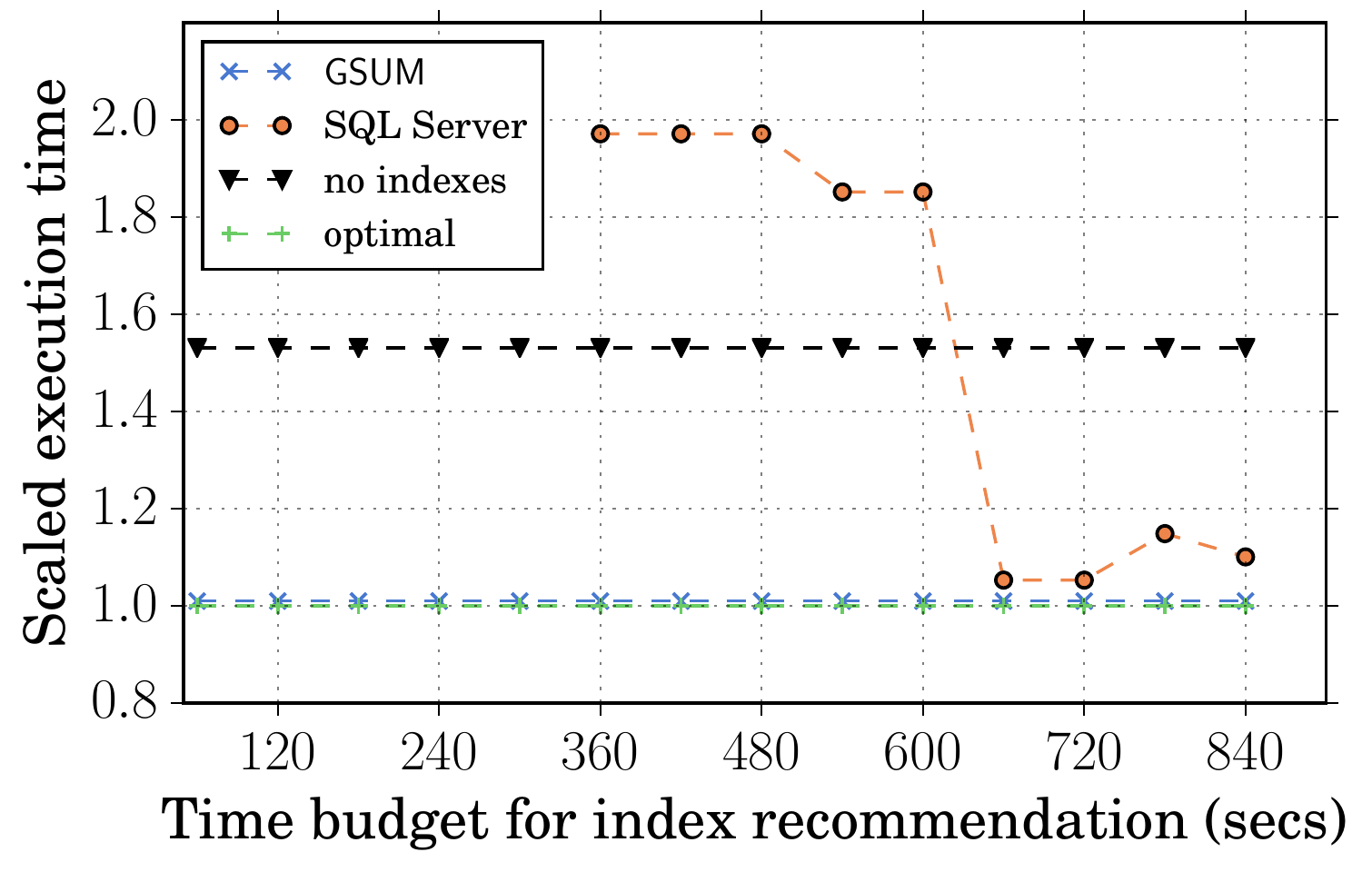}
		\caption{\textsf{TPC-H} workload} \label{fig:index:tpch:uniform}
	\end{subfigure}
	\hspace{0.5em}
	\begin{subfigure}{0.32\linewidth}
		\includegraphics[scale=0.35]{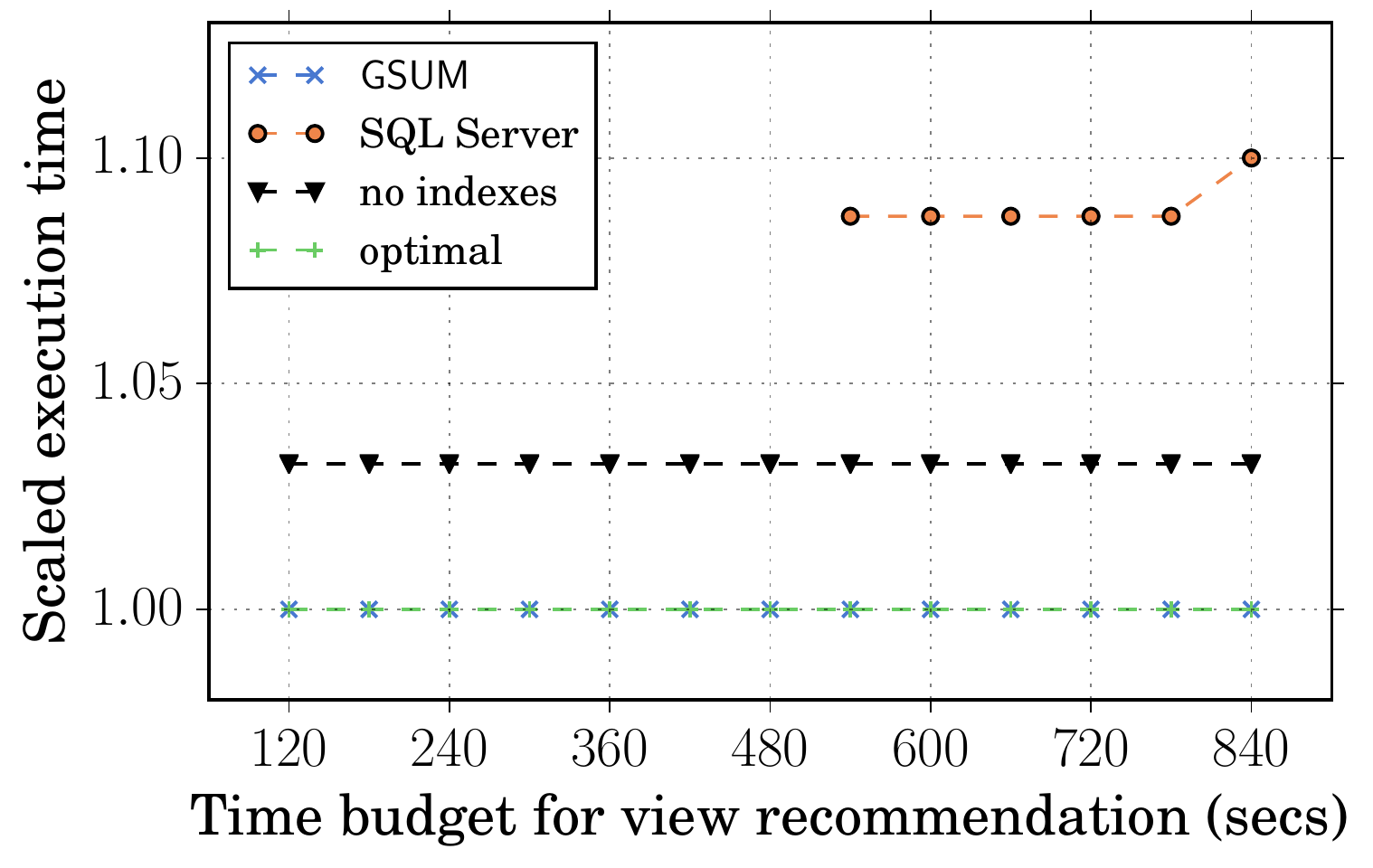}
		\caption{\textsf{TPC-H} workload} \label{fig:index:tpch:uniform:view}
	\end{subfigure}
	\hspace{0.5em}
	\begin{subfigure}{0.32\linewidth}
		\includegraphics[scale=0.35]{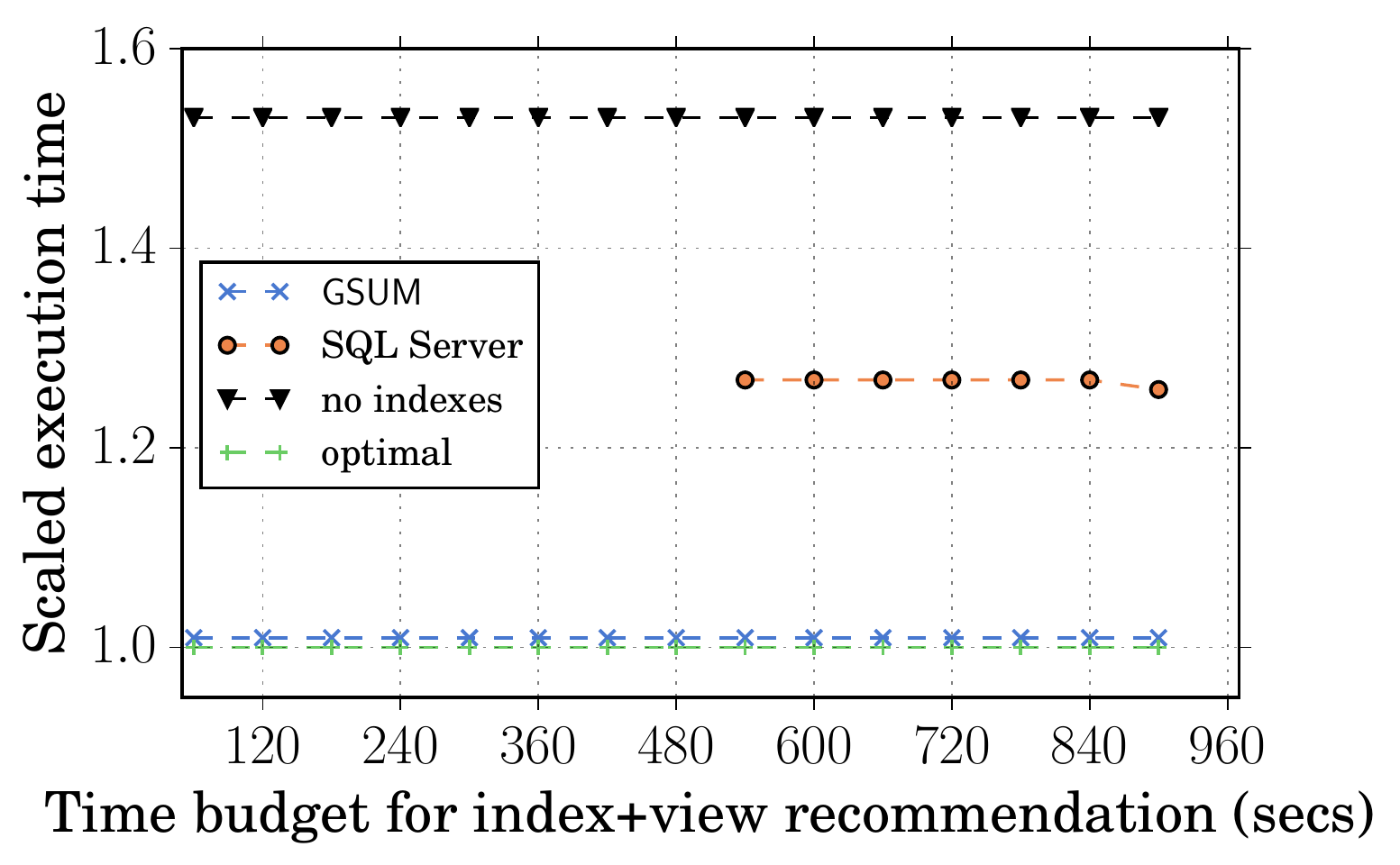}
		\caption{\textsf{TPC-H} workload} \label{fig:indexandview:tpch}
	\end{subfigure}

	\begin{subfigure}{0.32\linewidth}
		\includegraphics[scale=0.35]{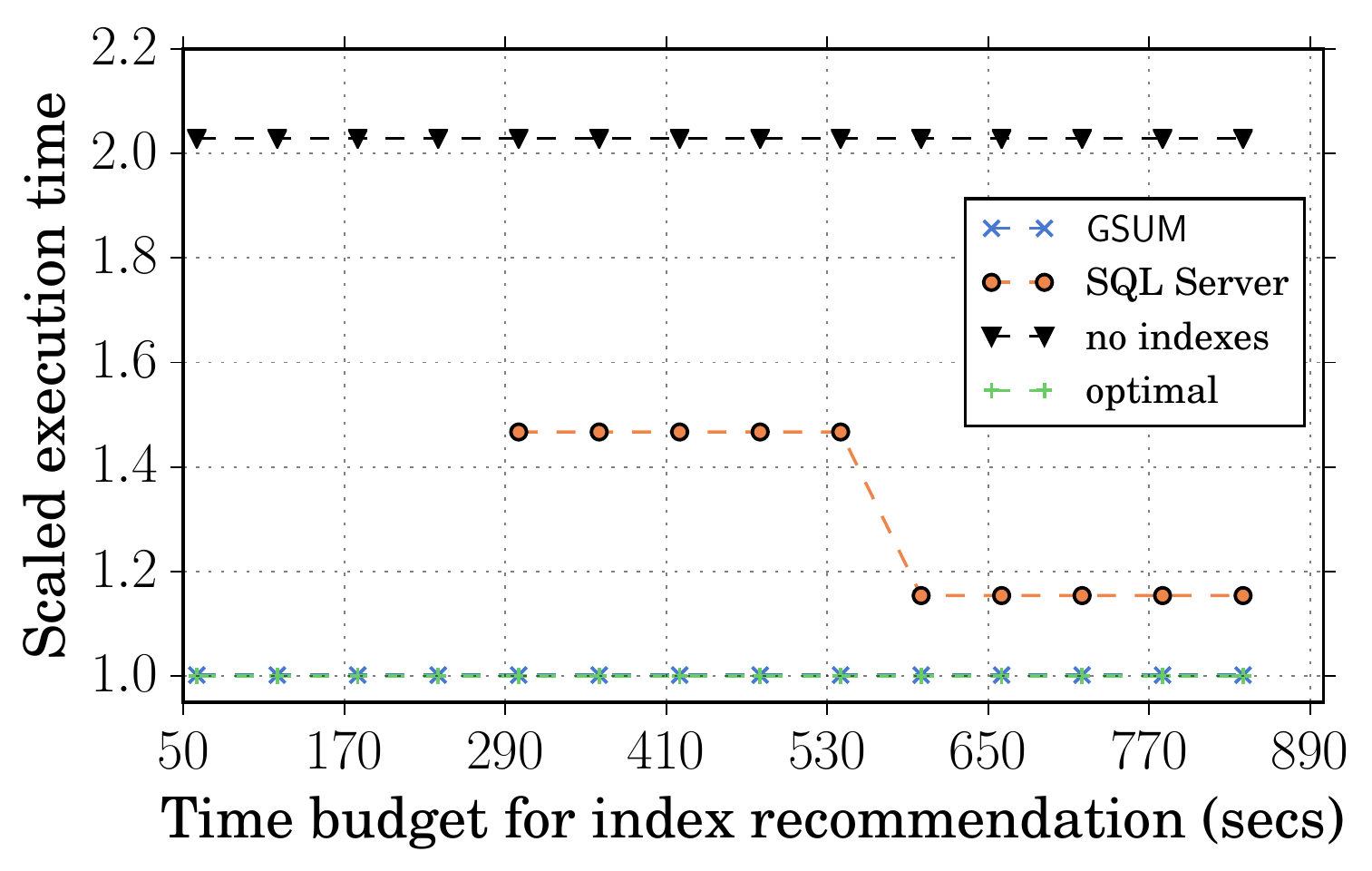}
		\caption{\textsf{SSB} workload} \label{fig:index:ssb:uniform}
	\end{subfigure}
	\hspace{0.5em}
	\begin{subfigure}{0.32\linewidth}
		\includegraphics[scale=0.35]{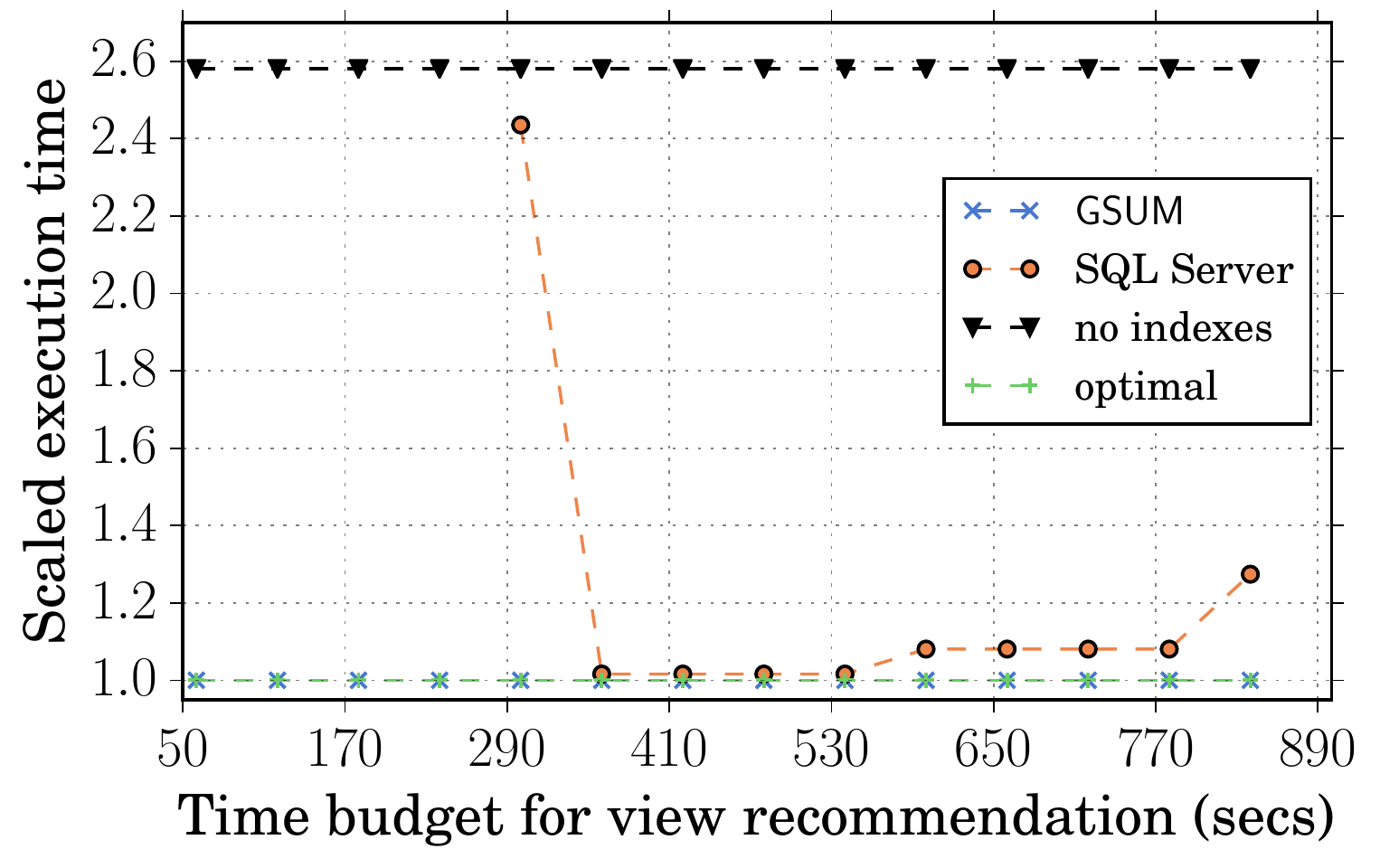}
		\caption{\textsf{SSB} workload} \label{fig:index:ssb:uniform:view}
	\end{subfigure}
	\hspace{0.5em}
	\begin{subfigure}{0.32\linewidth}
		\includegraphics[scale=0.35]{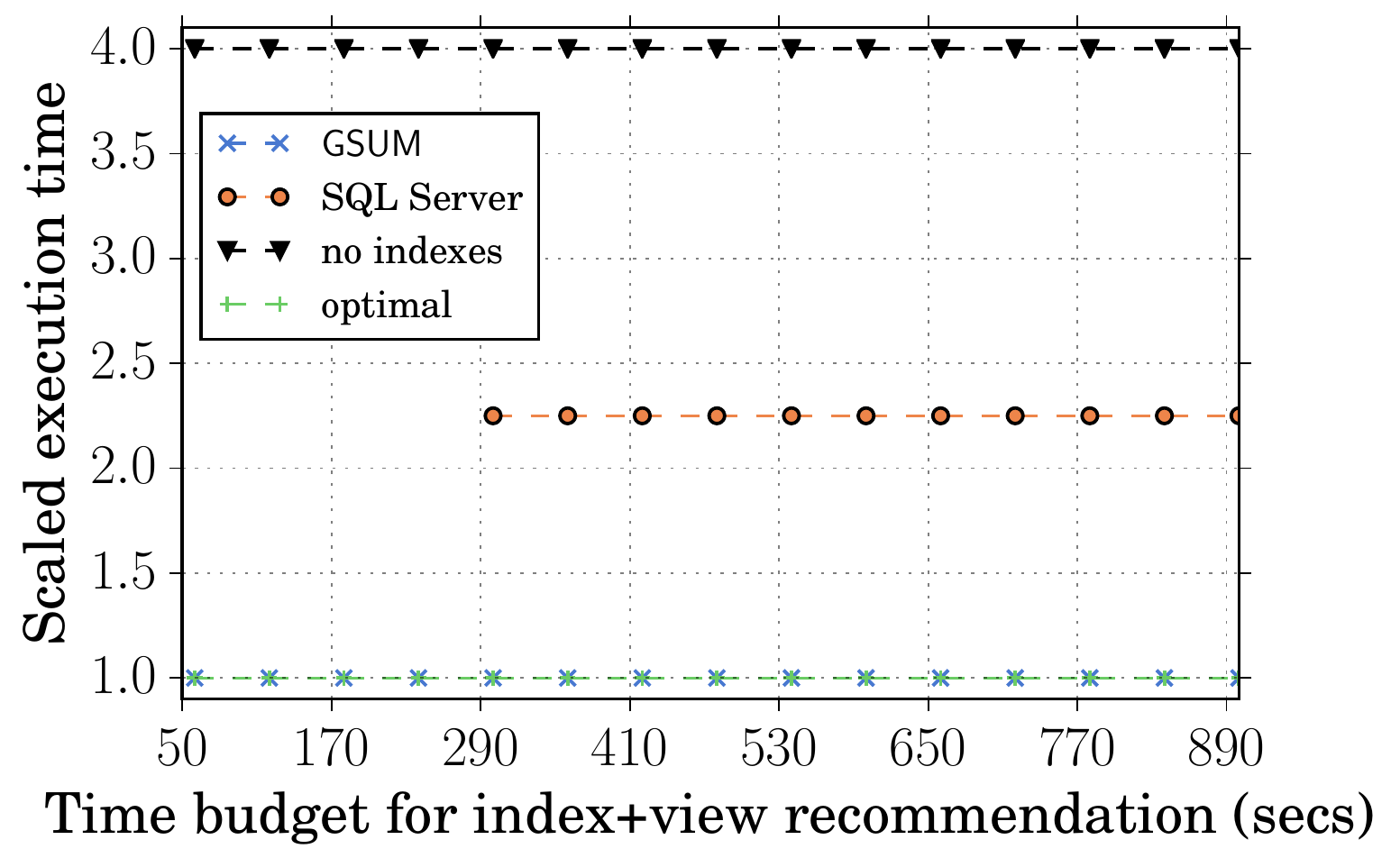}
		\caption{\textsf{SSB} workload} \label{fig:indexandview:ssb}
	\end{subfigure}
	\caption{Experimental results for three use cases: index tuning (a,d), materialized views (b,e) and indexes and indexed views (c,f). Execution time is scaled using the total running time when using optimal indexes to show the comparative slowdown.}
\end{figure*}

\rtwo{For all workloads, we assume that the query log is available through the DBMS.} All running time related experiments report the mean of the three observations that are closest to the median over a total of five runs. 
To normalize the numeric features, we set $H = 1000$. Unless specified otherwise, we choose $\gamma \rightarrow 0$ and $\beta = 0.5$. We refer to the compressed workload generated by our technique as \framework (short for \textsf{G}oogle \textsf{SUM}marized workload). We perform our experiments over $3$ datasets:

\begin{packed_enum}
	\item \texttt{DataViz}: A dataset of $512796$ ad-hoc data visualization queries issued against F1. 
	This workload contains references to $2729$ relations in total. 
	The largest join query contains $19$ joins and the workload has $106$ unique function calls. Most expensive query in the workload takes $6$ hours to execute.
	\item \texttt{TPC-H}~\cite{tpch}: A benchmark for performance metrics over systems operating at scale. We use a workload of $2200$ queries with \textsf{SF=1} and uniform data distribution.
	\item \texttt{SSB}~\cite{o2007star} : A benchmark designed to measure performance of database in support of  data warehousing applications.
\end{packed_enum}

\subsection{Use cases} \label{sec:use:cases}
As described earlier in this paper, there are several use cases for a summarization framework.
We now explore three of these use cases, index tuning, materialized view recommendation and both of them together, to show the validity of our framework and demonstrate that coverage and representativity metrics are useful in practice. Due to a lack of space, we defer the experiments for test workload generation to the appendix.

\smallskip
\johnc{\introparagraph{Experimental Setup} We use the SQL Server DTA utility as the baseline, which is a state-of-the-art industrial-strength tool that has been shipped in SQL Server for more than a decade~\cite{agrawal2005database}. All experiments in this section are run on a $m5a.8xlarge$ \textsf{AWS} \textsf{EC2} instance using a single core.}

\subsubsection{Index Tuning}
Index tuning is the task of selecting appropriate indexes for a workload that improve its overall runtime.
Summarization can be used in this context to determine a subset of relevant queries from the input workload and then generating indexes from the subset rather than the whole workload.

\smallskip
\introparagraph{Methodology} 
To evaluate summarization in the context of index tuning, we leverage the same evaluation strategy as Chaudhuri et al.~\cite{chaudhuri2003primitives} and subsequently used by Jain et al.~\cite{jain2018query2vec}.
That is, we first measure the execution time of a workload without indexes ($t_{orig}$) and then apply an index recommendation engine, determine the recommended indexes, create these indexes, and then measure the runtime again ($t_{sub}$).
As a baseline, we use SQL Server 2016, which comes with a built-in Database Engine Tuning Advisor.
Our experiments for SQL Server show our measurements for $t_{sub}^{SQL}$ under different temporal budget constraints for the tuning advisor, i.e.,~we vary the time allotted to the advisor that determines the indexes to create.
For comparison, we run \framework to create a summary workload to generate $\bS$ with constraint $|\bS| \leq \sqrt{|\bW|}$ while maximizing representativity and coverage, which we then use as input for the SQL Server tuning advisor.
Using these indexes, we can then measure $t_{sub}^{GSUM}$, i.e.,~the time it takes to run the input workload while using index suggestions based on the compressed workload generated by our algorithm.
We use all categorical and numeric features described in Section~\ref{sec:framework}.

\smallskip
\introparagraph{Results} 
\Cref{fig:index:tpch:uniform} to \Cref{fig:index:ssb:uniform} show the execution time of varying benchmarks given different index recommendation budgets with two baselines: {\em no indexes} and {\em optimal}, which uses indexes based on the 22 \textsf{TPC-H} query templates.
Diving deeper into \Cref{fig:index:tpch:uniform}, we note that under tight index generation budget constraints (between 6 and 11 minutes), SQL Server may give recommendations that result in worse performance than using no indexes at all.
The reason is that the \textsf{TPC-H} workload is large, so the advisor is unable to recommend appropriate indexes within these time constraints.
We further note that the tuning advisor's behavior is non-monotonic, which explains an increase in $t_{sub}^{SQL}$ with a larger index recommendation time budget.
In contrast, using \framework results in a much smaller workload for the advisor to interpret, reducing the time it takes to find index suggestions significantly.
With \framework, we can obtain the first suggestion for indexes within $1$ minute, while it takes SQL Server $6$ minutes to derive its first result.
SQL Server suggests same indexes as \framework at $11$ minutes. However, due to its non-monotonic behavior, an additional time budget worsens $t_{sub}^{SQL}$ marginally.
%Note that {\em optimal} is slightly better than \framework because the compressed workload does not contain $Q_1$ and $Q_{16}$, as other queries cover their features.
For \textsf{SSB} (\Cref{fig:index:ssb:uniform}), we observe that that the first index is recommended after $5$ minutes which is subsequently improved when the budget is $10$ minutes. \rtwo{We also verified that even after running the full workload with a budget of $60$ minutes, the recommended indexes were no better than the indexes recommended under $1$ minute. 
This demonstrates the benefit of using a compressed workload as opposed to the full workload. 
For both \textsf{TPC-H} and \textsf{SSB} workloads, the compression ratio is $\eta > 0.95$ since the compressed workload is always between size $10$ to $50$.}

\cut{Next, we compare our compression technique against index recommendation based on random sampling.
In this experiment, we vary the size of the summary workload.
As above, each of the output summaries is used for index recommendation and subsequent workload execution, resulting in a comparison of $t_{sub}^{GSUM}$ with $t_{sub}^{RAN}$, i.e.,~the execution time of the workload based on random sampling.
We observe in \Cref{fig:index:tpch:uniform:varysize} that \framework always outperforms random sampling independent of the summary size.
Moreover, after $k=30$, the performance of random sampling starts to degrade. The reason for this degradation is that choosing summaries for $30 < k < 39$ leads to less representativity as some query templates are present twice as often, which introduces skew whereas the input distribution is uniform. 
}

\begin{table*}[t]
	{\scriptsize
	\scalebox{0.90}{
		\begin{tabular}{ c|c|c|c|c|c|c|c}
			\toprule[0.1em]
			Task $\downarrow$ Feature $\rightarrow$  & \texttt{execution\_time} &  \texttt{output\_size} &  \texttt{\#joins} &\texttt{\#joins}$+$\texttt{output\_size}&
			\texttt{execution\_time}$+$\texttt{\#joins} & \texttt{execution\_time}$+$\texttt{output\_size} & \texttt{all numeric}  \\ \midrule[0.1em]
			\textsf{SSB} index & $1.1 \times$ &  $1.23 \times$ & $1.23 \times$ & $1.23 \times$& $1.1 \times$ & $1 \times$ & $1 \times$\\ \midrule
			\textsf{SSB} views & $1.08 \times$ &  $1.31 \times$ & $1.58 \times$ & $1.29 \times$ & $1.08 \times$ &$1.09 \times$ & $1.09 \times$\\ \midrule
			\textsf{SSB} indexes+views & $1.13 \times$ &  $1.35 \times$ & $1.85 \times$ & $1.35 \times$ & $1.13 \times$ & $1.08 \times$ & $1.08 \times$ \\			
			\bottomrule
	\end{tabular}}}
\caption{Using Numeric Features: Slowdown compared to using all categorical and numeric features for compression} \label{table:features:numeric}
\end{table*}

\begin{table}[t]
	{\scriptsize
	\scalebox{0.90}{
		\begin{tabular}{ c|c|c|c|c}
			\toprule[0.1em]
			Task $\downarrow$ Feature $\rightarrow$  & \texttt{function\_call} &  \texttt{table\_reference} &  \texttt{group\_by} & \texttt{order\_by}  \\ \midrule[0.1em]
			\textsf{SSB} index & $1.56 \times$ &  $1.23 \times$ & $1.36 \times$ & $1.36 \times$  \\ \midrule
			\textsf{SSB} views & $1.88 \times$ &  $1.58 \times$ & $1.17 \times$ & $1.41 \times$ \\ \midrule
			\textsf{SSB} indexes+views & $3.2 \times$ &  $1.85 \times$ & $1.95 \times$ & $1.5 \times$ \\			
			\bottomrule
	\end{tabular}}}
	\caption{Using Categorical Features: Slowdown compared to using all categorical and numeric features for compression} \label{table:features:categorical}
\end{table}

\subsubsection{Materialized view recommendation}
We utilize \framework to suggest materialized views using the inbuilt materialized view recommendation tool of SQL Server.

\smallskip
\introparagraph{Results} 
Figure~\ref{fig:index:ssb:uniform:view} and~\ref{fig:index:tpch:uniform:view} show the results for \textsf{SSB} and \textsf{TPC-H}.
For \textsf{SSB}, using \framework results in materialized views within $1$ minute that are better up to $2.5\times$ faster than the {\em no indexes} baseline. For both the workloads, even allowing up to $15$ minutes of time does not improve the recommended views. After $15$ minutes, SQL Server recommends the same views as \framework. 
\rtwo{In fact, for some time budgets, using the recommended views is slower as evidenced by the increasing execution time for SQL Server for both \textsf{TPC-H} and \textsf{SSB}. }

\rtwo{\subsubsection{Index and view recommendations}
Finally, SQL Server allows a third setting where both indexes and views can be recommended together. 
This usually allows for more indexes over the recommended views that further improve the workload performance. 

\smallskip
\introparagraph{Results} For \textsf{SSB} and \textsf{TPC-H} datasets, the performance is $2.5\times$ and $1.5 \times$ better respectively when using compressed workloads from \framework\ by using $5\times$ lesser time for generating recommendations as compared to when using the full workload. 
For \textsf{TPC-H}, we observed that even after $30$ minutes, the recommendations produced using the full workload are no better than the recommendation generated after $9$ minutes. We also observed the non-monotonic behavior of SQL Server tuning advisor when using the full workload. For some values of tuning time close to $30$ mins, the indexes and recommended views increases the workload execution time, again highlighting that allowing more time does not necessarily improve the quality of recommendations, further evidencing the advantage of using \framework.}

\begin{figure*}[t]
	\begin{subfigure}{0.46\linewidth}
	    \includegraphics[scale=0.40]{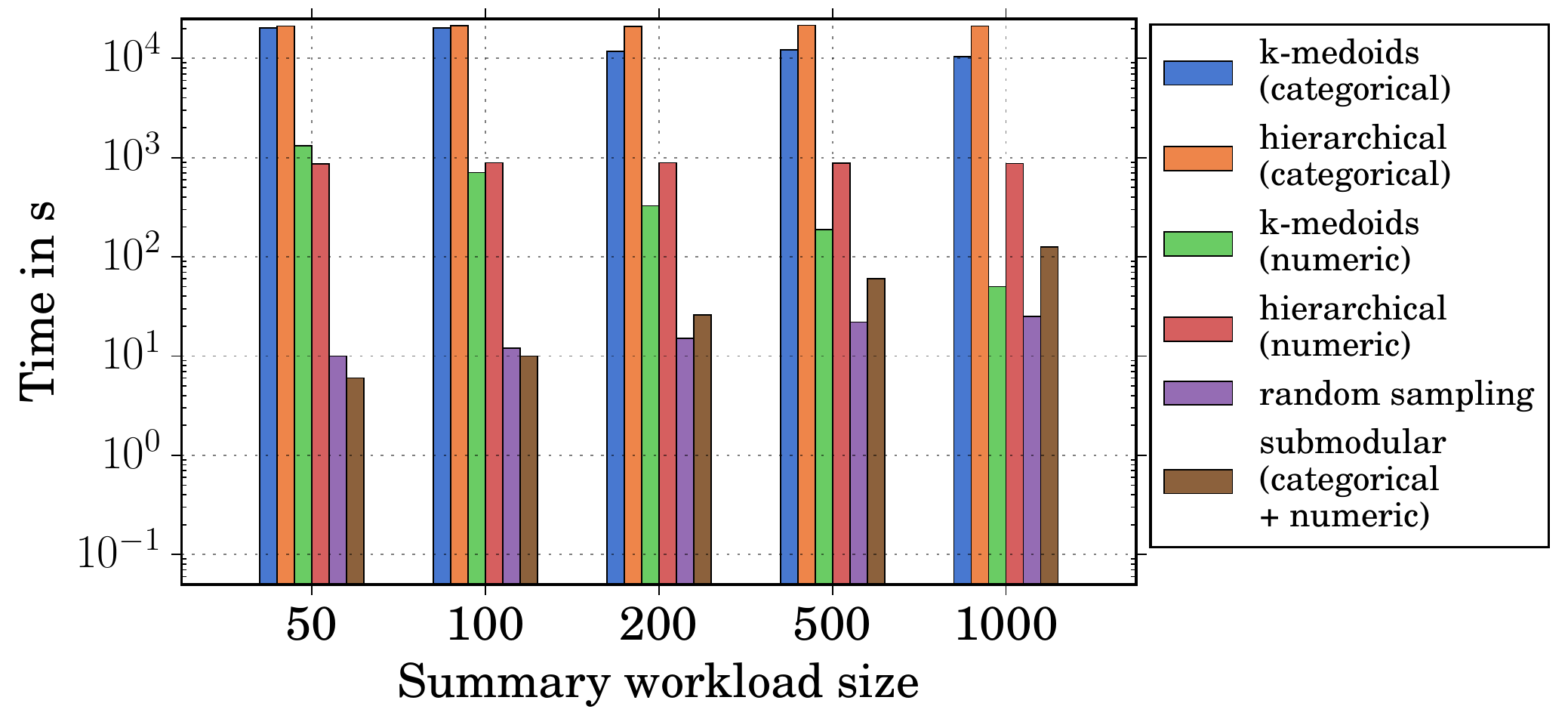}
		\caption{Runtime} \label{fig:timetaken}
	\end{subfigure}
	\hspace{2em}
	\begin{subfigure}{0.46\linewidth}
	    \includegraphics[scale=0.40]{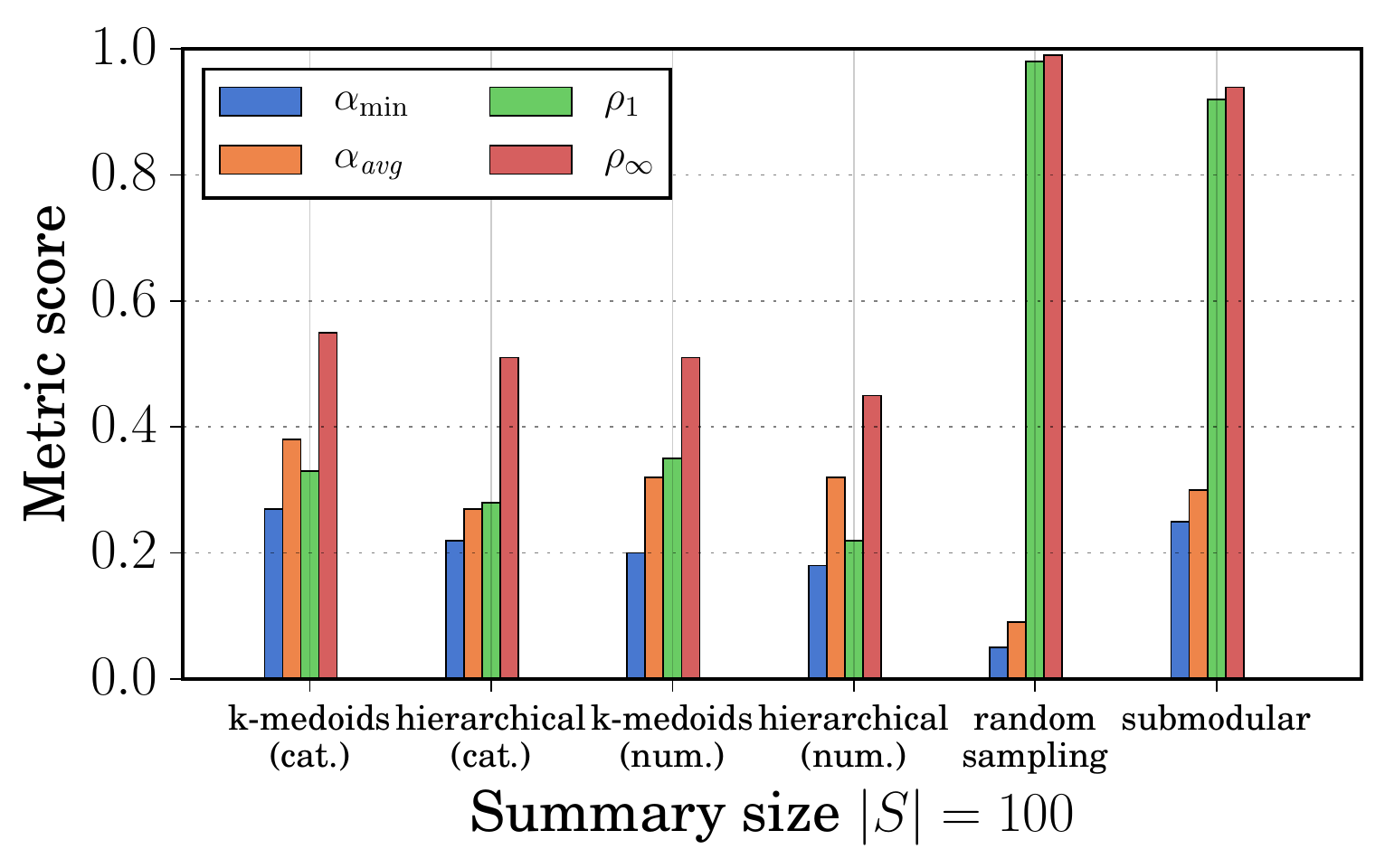} 
	    \caption{Coverage and representativity scores ($|\bS| = 100)$.} \label{fig:repcov}
	\end{subfigure}
	\caption{Runtime and metric scores for varying algorithms with $|\bW| = 5000$ on \texttt{DataViz}.}
\end{figure*}

\rthree{\subsection{Impact of Features} \label{sec:features}
Next, we perform an ablation study to see the impact of using a subset of features for the purpose of compression and compare the quality of the compressed workload \johnc{using the same experimental setup as in the previous section}. The metric we will use is slowdown in workload execution time using the indexes and views recommended using the compressed workload obtained with subset of features vs using all categorical and numeric features.~\autoref{table:features:numeric} and~\autoref{table:features:categorical} show the impact of using a subset of features for compression. The first observation is that all features (except \texttt{execution\_time}) when used in isolation fail to identify a good compressed workload. \texttt{function\_call} performs the worst because all queries in the \textsf{SSB} workload contain exactly one function call ({\color{blue} \textsf{SUM}}) which gives no useful information to the summarization task. The \texttt{table\_reference} feature is able to extract $4$ templates from the total $13$ \textsf{SSB} templates, while the \texttt{order\_by} and \texttt{group\_by} features extract the largest number of templates from the input workload. We remark that even when using all categorical features together, the average slowdown for all three tasks is $1.4 \times$. On the other hand, the numeric feature \texttt{execution\_time} performs very well and is able to identify almost all templates. This is because even for queries with high syntactic similarity (ex. $Q1.1, Q1.2, Q1.3$), the execution time varies enough for the algorithm to identify that they originate from different templates. This demonstrates the importance of numeric features in the algorithm. Using \texttt{execution\_time} with \texttt{output\_size} further improves the compressed workload quality slightly. However, using \texttt{output\_size} with \texttt{\#joins} does not perform well. Our conclusions for \textsf{TPC-H} are similar and we defer the experiments to the appendix.}

\subsection{Microbenchmarks} \label{sec:exp}
In this section, we study the impact of different algorithms on the metrics, explore the effect of parameter $\gamma$ on the objective function, study the impact of algorithmic optimization, and provide empirical evidence of \framework's scalability.

\smallskip
\johnc{\introparagraph{Experimental Setup} For all experiments in this section, we have implemented our summarization framework on top of the F1 database~\cite{shute2013f1} within Google. It consists of two distinct modules: the featurization module that  transforms and materializes the feature vectors of all queries that have been executed on the DBMS, and the summarization module that uses the materialized feature vectors and the input from the user to generate the workload summary. We use the \texttt{DataViz} dataset as the basis for our comparison, as it is a representative production workload. All experiments in this section are executed on a single machine running Ubuntu $16.04$ with $60$GB RAM and $12$ cores.}

\subsubsection{Algorithm Comparison} \label{sec:algo}
In the previous sections, we have compared against an industry system.
However, there are several other algorithms that can be used in the context of workload summarization such as clustering techniques.
To examine these, we have implemented k-medoids and hierarchical clustering (average linkage) in addition to random sampling and \framework.

\smallskip
\introparagraph{Methodology}
In this set of experiments, we use $5000$ randomly chosen \texttt{DataViz} queries\footnote{A small sample is chosen to make sure that clustering algorithms can finish running.} and vary the summary workload between $50$ and $1000$.
We compute the Euclidean distance over numeric features and the Jaccard distance over categorical features as distance function for clustering. 

\smallskip
\introparagraph{Results}
Figure~\ref{fig:timetaken} shows the runtime of the different  algorithms. 
Our first observation is that using clustering algorithms with categorical features is the most expensive choice (time $> 10,000$ seconds). 
This is because even though the feature vectors are materialized, finding the Jaccard distance takes $O(\norm{\bq})$ time, as compared to $O(1)$ for the distance computation for single-valued numerical features. 
Further, in order to find the representative of each cluster, the number of operations required is quadratic in the cluster size, which amplifies the performance difference. 
%We used optimizations such as caching the distance between every pair of queries once it is computed and pre-computing the distance matrix (which is already $O(n^2)$) to improve execution time, which resulted in a runtime improvement of one order of magnitude.
%However, the asymptotic quadratic behavior of clustering does not change. 
%Specifically, increasing the size of the input workload to $15000$ is sufficient to offset the gains introduced by our optimizations. 
Our second observation is that as the summary size increases, the execution time decreases for k-medoids, it increases for \framework, and stays approximately the same for hierarchical clustering and random sampling. Finally, we observe that compared to the two clustering algorithms whose performance depends heavily on the type of features used, the submodular algorithm is less sensitive to the chosen features since the submodular gain computation depends on a single query.
%For k-medoids, we see that if $\bS$ is small, most queries are concentrated in a few clusters.
%As the summary size increases, cluster skew reduces, which improves the runtime.
%In contrast, \framework's execution time increases because we need to make more passes over the input workload. 

\begin{figure}
	\includegraphics[scale=0.95]{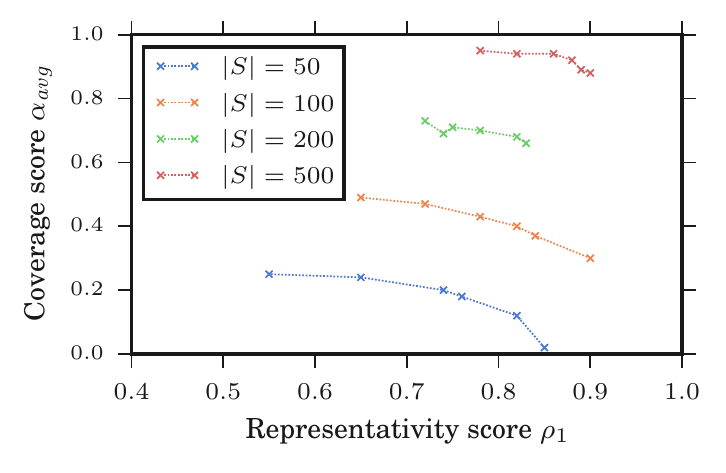}
	\caption{Trade-off between metrics; $\gamma$ is the smoothing parameter that varies between $0$ and $1$ for each curve.} \label{fig:smooth}
\end{figure}

\subsubsection{Representativity and Coverage} \label{sec:rep:cov} We now compare the metrics for the same workload and fix $\td(\cdot)$ to be the input distribution.

\smallskip
\introparagraph{Results}
\Cref{fig:repcov} shows the resulting coverage and representativity scores for all algorithms and summary size $|\bS| = 100$.
As expected, random sampling has the best representativity score and the lowest coverage. 
All clustering algorithms have low representativity scores but achieve good coverage. 
This is not surprising because the cluster initialization step identifies {\em outlier} queries as cluster centers since they have the largest distance from most other queries. 
In comparison, \framework has a slightly lower coverage score than the best clustering algorithms, but performs significantly better in terms of representativity. 
Note that since $\gamma \rightarrow 0$, the submodular algorithm focuses on maximizing coverage first and the chosen input workload of $100$ queries is not able to cover the active domain (as are the clustering techniques). 
Once \framework has reached the best possible coverage for a fixed $\gamma$, $\rho$ improves even further. 
\rtwo{Given the faster running time of random sampling, it is natural to ask why \framework is better than random sampling. 
Note that random sampling provides poor coverage, since in the presence of skew, outlier queries are likely missed from the sample. 
These drawbacks were also identified by~\cite{chaudhuri2002compressing,chaudhuri2003primitives}. 
Secondly, it is not clear how to incorporate user-specified constraints, e.g.,~the execution time of the summary is at most $1$ hour, or a custom target distribution. 
Finally, the right sample size for random sampling is \emph{unknown} a priori. 
Similar to the clustering methods, random sampling needs the sample size as the input. 
This means that we may need to run random sampling for all possible sample sizes which makes it a less attractive choice.}

\begin{figure*}[t]
	\begin{subfigure}{0.3\linewidth}
		\includegraphics[scale=0.35]{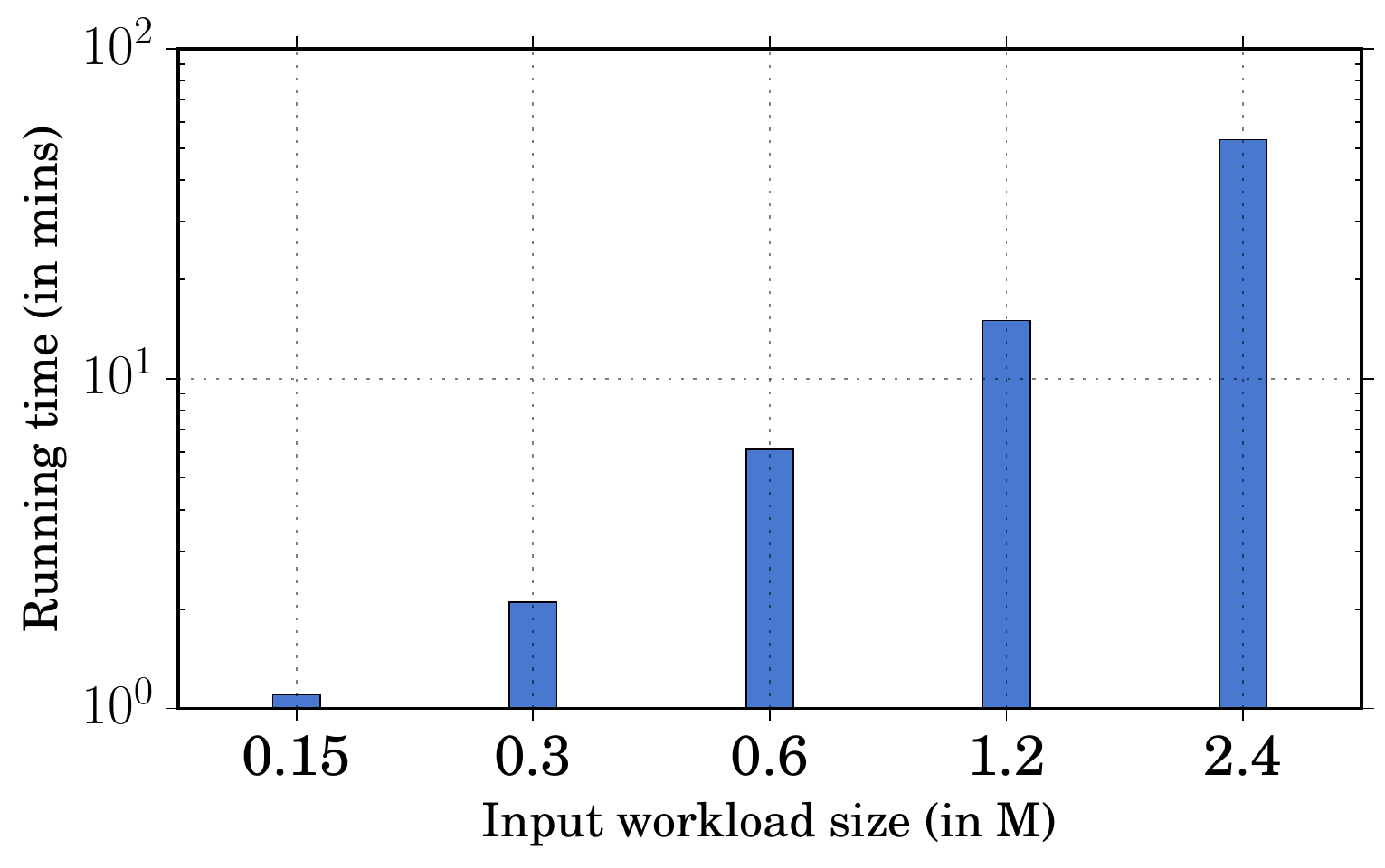}
		\caption{Runtime varying $|\bW|$, $|\bS|=\sqrt{|\bW|}$} \label{fig:scalability:exp1}
	\end{subfigure}
	\hfill
	\begin{subfigure}{0.3\linewidth}
		\includegraphics[scale=0.35]{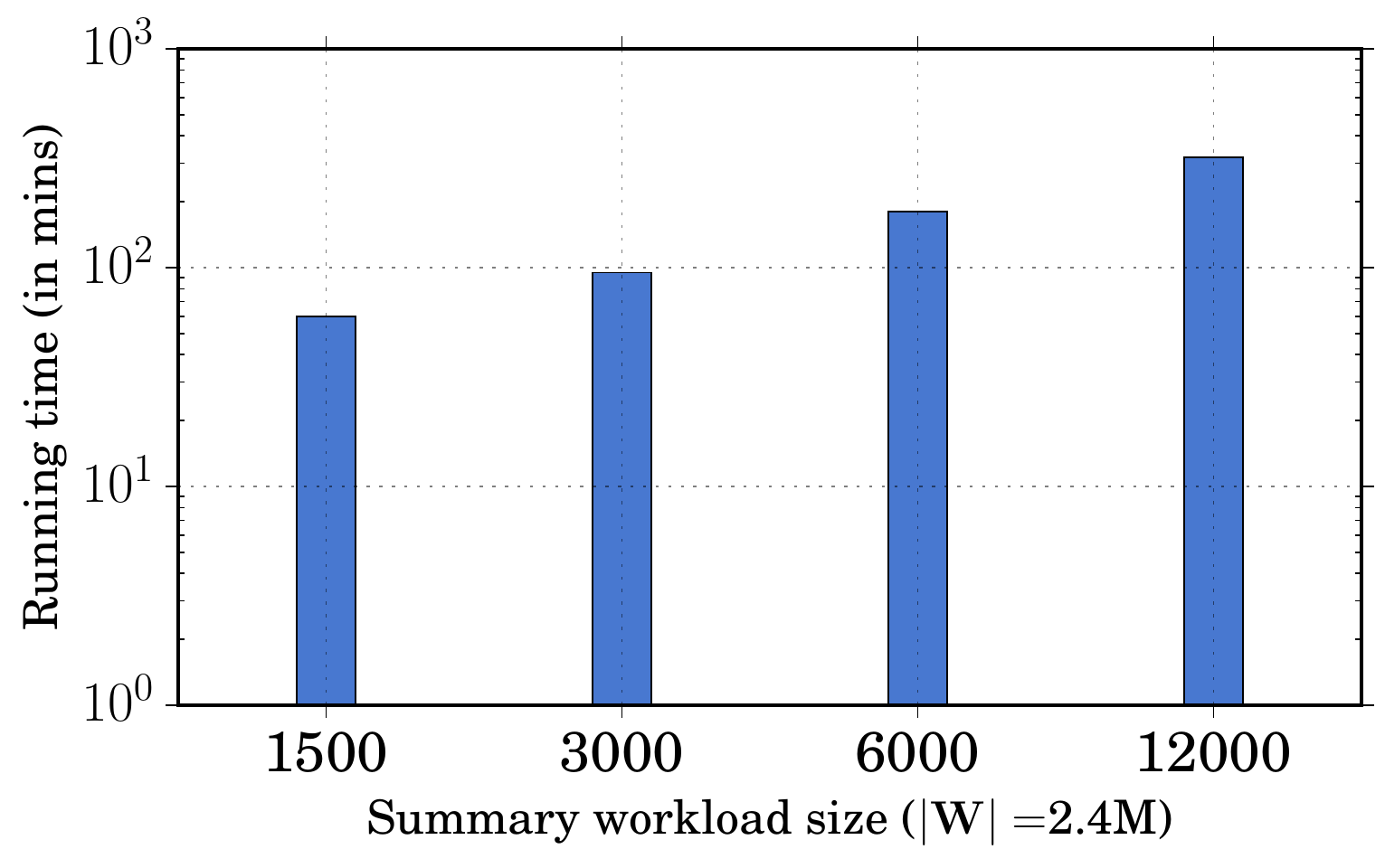}
		\caption{Runtime varying $|\bS|$, $|\bW|=2.4$M} \label{fig:scalability:exp2}
	\end{subfigure}
	\hfill
	\begin{subfigure}{0.3\linewidth}
		\includegraphics[scale=0.35]{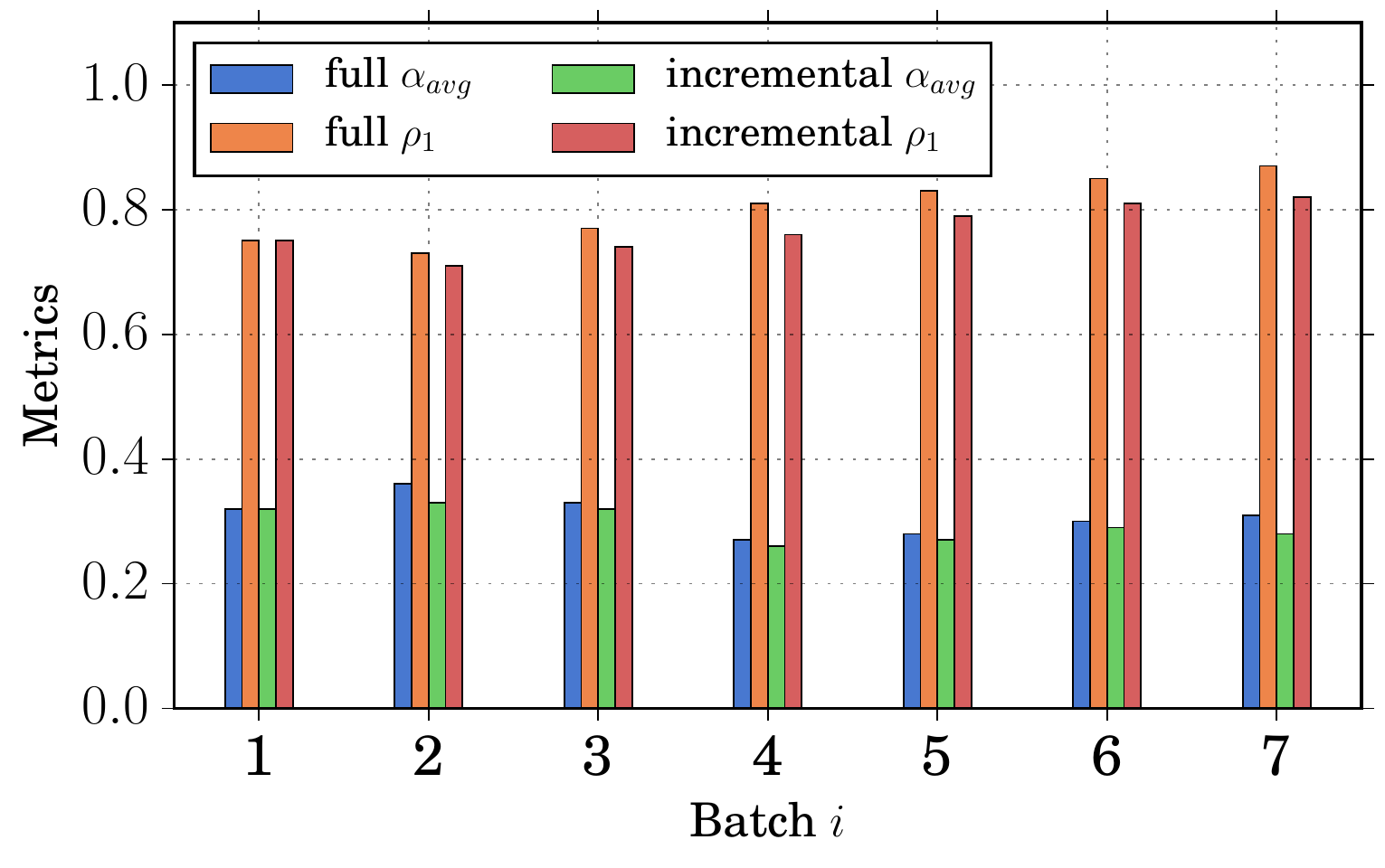}
		\caption{Coverage and representativity scores} \label{fig:incremental:dataviz}
	\end{subfigure}
	\caption{Scalability (a),(b), and incremental computation (c) experiments on \texttt{DataViz}.}
\end{figure*}

\subsubsection{Trade-off between Metrics}
In the next experiment, we empirically verify the impact of the $\gamma$ parameter on \textsf{DataViz} workload and study how it can be used to trade-off between the two metrics. 

\smallskip
\introparagraph{Results}
\Cref{fig:smooth} shows the trade-off between coverage and representativity (by controlling $\gamma$) for different summary sizes. 
Observe that as $\gamma$ decreases, representativity starts dropping and coverage starts to increase. 
Note that the increase in coverage or decrease in representativity is not necessarily monotone due to the complex interaction between features.
Since all features are uniformly weighted, features with more tokens tend to dominate both $\alpha$ and $\rho$. 
In order to balance that, we can set the weight for each feature inversely proportional to its active domain. 
However, in almost all our experiments, no feature dominated a different feature even when each feature had the same weight. 
In other words, the representativity score for each feature shows an empirical monotone behavior as $\gamma$ changes.
\Cref{tradeoff:prod} shows the impact of $\gamma$ on the full \texttt{DataViz} workload. With $|\bW| = 512796$ and $|\bS| = 1000$, we observe that as the value of $\gamma$ decreases, $\rho_{1}$ increases at the expense of $\alpha_{\mathit{avg}}$. 
Using the generated summary, we are able to identify recurring patterns in the input workload. Finally, we observe that the adjustments to $\gamma$ depend on the skew of the workload and should be examined on a case-by-case basis.

 \begin{table}[t]
 	\scalebox{0.99}{
 		\begin{small}
 			\begin{tabular}{@{}lrrrrrr@{}}\toprule
 				\textbf{\# processors $\rightarrow$} & $\mathbf{1}$&$\mathbf{2}$&$\mathbf{3}$&$\mathbf{4}$&$\mathbf{5}$&$\mathbf{6}$  \\ \toprule
 				\textbf{runtime (min)} & 95 & 53 & 36 & 26 & 21 & 18 \\ \midrule
 				\textbf{speedup obtained} & 1x& 1.8x& 2.7x & 3.6x & 4.5x & 5.5x \\\midrule
 				\textbf{coverage} & 1.0 & 1.0 & 1.0 & 1.0 & 1.0 & 1.0 \\\midrule
 				\textbf{representativity} & 0.91 & 0.88 & 0.86 & 0.85 & 0.85 & 0.85 \\
 				\bottomrule
 			\end{tabular}
 		\end{small}
 	}
  	\caption{Runtime (in minutes) and metrics obtained using parallel processing.}
 	\label{table:runtime}
 \end{table}

\begin{table}[t]
	{\scriptsize
	\scalebox{0.95}{
		\begin{tabular}{ c|c|c|c|c|c}
			\toprule[0.1em]
			$\gamma$ & $f^s_1$ &  $f^p_1$ &  $f^p_2$ & $f^p_4$ & $f^p_5$  \\ \midrule[0.1em]
			$10^0$ & $ 0.75 \mid 1.0 $ & $ 0.67\mid 1.0$ & $ 0.57\mid 1.0$ & $0.9\mid 1.0$ & $0.5\mid 1.0$  \\ \midrule
			$10^{-3}$ & $ 0.76\mid 1.0$ & $0.74\mid 0.81$ & $ 0.66\mid 0.88$  & $0.90\mid 1.0$ & $0.6\mid 0.84$  \\		\midrule
			$10^{-10}$ & $ 0.8\mid 1.0$ & $ 0.81\mid 0.68$ & $ 0.73\mid 0.75$ & $ 0.90\mid 1.0$ & $0.69\mid 0.70$  \\	\midrule
			$10^{-15}$ & $ 0.84\mid 1.0$ & $0.87\mid 0.57$ & $0.80\mid 0.61$ & $ 0.95\mid 0.56$ & $ 0.74\mid 0.64$  \\	\midrule
			$10^{-20}$ & $ 0.88\mid 1.0$ & $0.98\mid 0.32$ & $0.94\mid 0.26$ & $ 0.98\mid 0.30$ & $ 0.94\mid 0.18$  \\	\midrule
			$10^{-25}$ & $ 0.91\mid 1.0$ & $0.99\mid 0.07$ & $0.94\mid 0.07$ & $ 0.98\mid 0.07$ & $ 0.99\mid 0.02$  \\
			\bottomrule
	\end{tabular}}}
	\caption{Tradeoff between metrics for production workload. Each cell shows $\rho_{1} \mid \alpha_{\mathit{avg}}$} \label{tradeoff:prod}
\end{table}

\subsubsection{Scalability and Parallel Computation}

To benchmark scalability, we execute \framework on a single thread on a single machine and use all available categorical and numeric features. 
\rtwo{We use a workload consisting of $2.4M$ \textsf{TPC-H} and \textsf{SSB} queries.}

\smallskip
\introparagraph{Results}
\Cref{fig:scalability:exp1} shows the runtime in minutes when the input workload size $|\bW|$ varies and the summary size is fixed to $|\bS| = \sqrt{|\bW|}$. If $|\bW| = 2.4$ million queries, it takes \framework $53$ minutes to execute compression. 
\rtwo{Since the compression algorithm is executed daily, this performance is acceptable in practice. 
As we will see later, using multiple cores can further improve the runtime.}
We generally observe a linear increase in runtime when increasing the workload size. \Cref{fig:scalability:exp2} shows how the choice of summary size impacts scalability.
Here, we set $|\bW| = 2.4$ million queries and observe that creating a summary with $|\bS| = 12000$ takes roughly $5.2$ hours.
Analogous to the results observed when increasing $|\bW|$, we see a linear increase in runtime with an increase in summary size $|\bS|$. \Cref{table:runtime} shows the runtime and the impact of parallelization on the summary workload metrics.  The first column shows the metrics when a single processor calculates the summary. As the number of processors increases, the speed-up obtained is near linear. We further observe no impact on coverage and representativity only drops marginally from $0.91$ to $0.85$.
 
\subsubsection{Incremental Computation} 
Recall that incremental computation computes a local summary for each batch of input queries and merges them, instead of recomputing the summary from scratch. 
To test this behavior, we split \texttt{DataViz} into \rtwo{$7$ different batches by partitioning the query log sequentially} and summarize the workload incrementally, adding one batch at a time.
We compare our incremental results to a from-scratch execution of \framework on the same subset of queries. 
For this experiment, we set $|\bS| = \sqrt{|\bW|}$. 

\smallskip
\introparagraph{Results}
\Cref{fig:incremental:dataviz} shows our findings.
With an increase in the number of processed batches, we observe incremental computation providing marginally worse results than creating a summary from scratch.
At the same time, we observe that the cumulative time spent in creating a summary from scratch ($\approx 480$ minutes) is much larger than the cumulative time of merging summaries across batches ($\approx 60$ minutes).
After receiving batch $i$, the total input workload size has increased by a factor of $i$ and the summary size by a factor of $\sqrt{i}$, increasing the runtime by a factor of $i^{3/2}$ as compared to $i-1^{\mathsf{th}}$ batch.
Using incremental computation, we can avoid this runtime increase as each batch is treated independently.

\subsection{Discussion} \label{sec:discussion}
\noindent \textbf{Choice of Features.} 
So far, we have observed how different knobs and configuration parameters impact various performance metrics and the output of \framework\ but we have not discussed how to choose them in practice. 
The answer to this question depends on $(i)$ the chosen application, and $(ii)$ how the chosen features interact with each other. 
While adding more features certainly provides more signals, it is not necessarily the case that it will guarantee a better result.
For instance, the addition of more features will require a larger summary size to obtain better coverage, which in turn can possibly decrease representativity. 
For our experiments, the choice of features was driven mainly by iterating over the available choices and understanding their impact. 
Two heuristics that were useful to us are: 
$(i)$ We found multiple production workloads that contained  $> 10^5$ table references, most of which were temporary tables. 
For such workloads, including \texttt{table\_reference} and attributes in {\color{blue} \textsf{WHERE}} clause as a feature is not a good idea; 
$(ii)$ if two features $f_1$ and $f_2$ are highly correlated, then it suffices to include only one of these features. 
We found it beneficial to perform multiple iterations by generating multiple workloads with different features and then look at the compressed workload to understand how the summary has changed by using the visualization tools present in \texttt{DIAMetrics}. We found it useful to change feature weights and introduce weights for each token (initially
uniform) and then change in each iteration to boost
metric scores. For features with large domain, setting $\gamma$ closer to $1$ to dampen the effect of coverage also worked well.
For tasks such as choosing platform alternatives to execute a workload, using categorical features is not necessary. Indeed, the execution performance of a workload has little do with SQL syntax and more to do with physical execution plans and operators available on different systems. 

\noindent \textbf{Choice of $\beta$.} The right choice of $\beta$ is determined by the application for which the compressed workload will be used. Applications that require outliers in the compressed workload (such as test workload generation and benchmarks for creating compliance benchmarks) set $\beta$ closer to $1$ whereas applications such as index recommendation require representative workloads. However, even within a specific application the optimal choice of $\beta$ can change. As an example, consider two index recommendation algorithms $A$ and $B$. $A$ may choose to recommend indexes that optimize the execution time of the most expensive queries first. In this case, the compressed workload must contain the most expensive queries in the workload. On the other hand, $B$ may choose to recommend indexes that benefit the common-case queries in the workload but ignore the uncommon long running queries. This example demonstrates the need for a formal specification of the application context that can be integrated into the compression algorithm. Currently, we choose $\beta$ by constructing multiple summaries and then test the performance to find the right threshold.

\section{Limitations and Future Work} 
\label{sec:limitations}

We now discuss limitations of our work and ideas that will drive the agenda for this line of research.

\smallskip
\noindent \textbf{Feature Engineering.} One limitation of our framework is that it only looks at features at {\em query level} but does not incorporate workload level features such as contention between queries for resources. 
Choosing the right set of features for each applications is also a bottleneck. Currently, our features are hand-picked by domain experts such as application and database developers, database administrators and support personnel. Finding the right set of features requires an iterative analysis of the query logs to understand the variability in feature values. 
Since \framework is much faster than clustering algorithms, it allows us to build multiple benchmarks with different sets of features and $\beta$. 
As a part of future work, we plan to utilize machine learning techniques to identify the best features for a given application.

\smallskip
\noindent \textbf{Transactional Workloads.} In this work, we focus only on analytical workloads, ignoring the impact of data. 
For transactional workloads, the runtime of a query changes as the skew in the data changes. Thus, we need more sophisticated techniques to construct compressed workloads that take data updates into account. 

\smallskip
\noindent \textbf{Auto-tuning Knobs.} Since the framework contains many knobs such as the  choice of $\beta$, budget constraint and different application contexts, it is worth exploring how we can find the optimal configuration of the knobs for each application. 
Deep reinforcement learning has had considerable success in performing this task.
	\section{Related Work}
\label{sec:related}
Most prior work focuses on maximizing coverage of information in summary workload as the primary optimization criteria while incorporating notions such as quality, efficiency as additional constraints. 
The property of representativity is more nuanced in comparison to coverage since it is highly dependent on the application context. 
In most cases, a representative summary minimizes the average distance from all items in the input workload~\cite{ku2006opinion, sarkar2009sentence, ranu2014answering}, maximize mutual information between summary and input~\cite{pan2005finding}, maximizes saturated coverage~\cite{mehrotra2015representative} or maximizes coverage and diversity~\cite{mirzasoleiman2016fast, tschiatschek2014learning, dasgupta2013summarization,eick2004using, srinivasan2009efficient, xu2014efficient}. 
In all these cases, the representative metric function is well behaved, i.e, it is monotonic and submodular by definition. 
Our problem setting departs from these works in our definition of representative where we would like the summary workload to mimic a target feature distribution. This makes the summarization problem more challenging. 
We note that our definition of representative has been proposed in previous work but has only been studied empirically as a quality metric whereas our algorithms are designed specifically to optimize for this metric.

\smallskip
\introparagraph{Compressing Workloads} Compressing or summarizing \textsf{SQL} workloads has been studied by Chaudhuri et al.~\cite{chaudhuri2003primitives, chaudhuri2002compressing}. 
In~\cite{chaudhuri2003primitives}, the authors propose multiple summarization techniques including K-Medoids clustering, stratified/random sampling and all pairs query comparison. 
As the authors themselves note, random sampling ignores valuable information about statements in the workload and  misses queries that do not appear often enough in the workload. 
~\cite{chaudhuri2002compressing} proposes a new \textsf{SQL} operator specifically for summarizing workloads but also suffer from the quadratic complexity. 
More recently,~\cite{8352666, jain2018query2vec} propose query structure based clustering methods for workload summarization but both proposed approaches are not scalable and rely only on the syntactic information in the text of the query.
In contrast, our framework is more general in the sense that we also incorporate query execution statistics. 
Ettu~\cite{kul2016ettu, kul2016summarizing} presents a promising approach that clusters workload queries based on query syntax but it has restrictive assumptions with respect to defining query similarity which is based only on the subtree similarities of the query syntax. 
It is geared towards clustering queries written by humans with the purpose of identifying queries that may constitute a security attack. 
While the framework is scalable, it has restrictive assumptions with respect to defining query similarity which is based only on the subtree similarities of the query syntax. 
This is not true in our setting where most queries are generated by a pipeline of processes that increase the size of the query making it difficult to find out if two queries have the same 'intent'.

\smallskip
\introparagraph{Counting Workload Patterns} An orthogonal but related problem to our setting is creating a compressed representation that allows for counting {\em patterns} in a query.
The key idea explored in~\cite{mampaey2012summarizing, xie2018query} is to develop a maximum entropy model over feature values that then allows us to query pattern counts by simply computing the product of probabilities that each feature value is present.

%\smallskip
%\introparagraph{Sampling based techniques} Random/Stratified sampling has been proposed as a strategy for workload summarization in the past~\cite{gibbons2002fast}. 
%However, the limitations of random sampling have also been noted by previous works~\cite{chaudhuri2003primitives, chaudhuri2002compressing} which are confirmed by us as well; the primary limitation being that sampling ignores the interaction between queries being sampled (which is also the reason for its efficiency). 
%For most applications, this performs sub-optimally since further opportunity compression (such as exploiting the sparsity) is lost. Also, sampling has additional limitations where satisfying user constraints and customizability is a problem.
%It remains an interesting problem to remedy these issues by considering biased/weighted sampling which will require new techniques to reason about statistical guarantees. 
	\section{Conclusion}
\label{sec:conclusion}

In this paper, we propose a novel and tunable algorithm that allows us to summarize workloads for various application domains. 
We show that the proposed solution provides provable guarantees and solves the underlying problem efficiently by exploiting its submodular structure. 
Our solution supports parallel execution as well as incremental computation model and is thus highly scalable.
We show through extensive experimental evaluation that our solution outperforms clustering algorithms and random sampling. We view this work as the first in an exciting research direction to develop automated solutions for workload monitoring at scale. We believe our solution can be extended
to tackle interesting problems such as (but not limited to)
resource prediction and production workload analysis that
can be applied for a variety of (database) systems.

\begin{acks}
	This research was supported in part by National Science Foundation grants CRII-1850348 and III-1910014. We are also grateful to the anonymous referees for their detailed and insightful comments that have greatly aided in improving this work.
\end{acks}
	
	\bibliographystyle{ACM-Reference-Format}
	\balance
	\bibliography{reference}

	\clearpage
	
	\appendix
	% !TeX root = paper.tex
\section{Proofs} \label{sec:a}

\begin{proof}[Proof of Theorem~\ref{thm:hardness}]
We will first prove the hardness result for $\alpha_{min}$ and $\rho_{\infty}$, and then extend it to the other combinations. 
Fix some $\beta \in [0,1)$; we already know that the problem is NP-hard for $\beta=1$. 

We will show NP-hardness by a reduction from a variant of the XSAT (exact SAT) problem. In XSAT, we are given a CNF formula such that every clause has size exactly $k \geq 3$, and each variable is present in exactly
$l \geq 3$ clauses. Further, each clause has only positive variables (so no negated variables). The goal is to find a satisfying assignment such that exactly one variable in each clause is true and all other variables in that clause are false. This problem is known to be NP-hard (\cite{porschen2014xsat}, Theorem 29).

Let the input instance for XSAT be a formula $\phi$ with $m$ clauses.
We construct a workload $\bW$, by introducing a query $\bq_x$ for each
variable $x$ that occurs in $\phi$. We consider a single multivalued feature $f$, where each token $t_c$ corresponds to a clause $c$.  For a query $\bq_x$, $f(\bq_x)$ contains exactly the clauses (tokens) that contain variable $x$. Since each clause contains exactly $k$ variables, we have that for any token $t$, $p_{\bW}(t) = k/(k \cdot m) = 1/m$. We set the target distribution to be
$q(t) = p_\bW(t)$, and choose the budget to be $B = m$. We claim that XSAT has a solution if and only if the quantity $\beta \alpha_{min} + (1-\beta) \rho_1$ has maximum $1$.

\smallskip
\noindent $\boldsymbol{\Rightarrow}$ Suppose there exists a solution to XSAT. We construct a summary $\bS$ by choosing the queries of the variables that are set to true. Since at most $m$ variables are true, we have $|\bS| \leq m$. Since for each clause (token) exactly one variable is true, the coverage for feature $f$ is  $\alpha_f = 1$, and thus $\alpha_{min} = 1$. Additionally,  each token appears exactly once in the solution, so the distribution is $p_{\bS}(t) = 1/m$. Hence, $\rho_1(q) = 1$. Summing up, for summary $\bS$ we have $\beta \alpha_{min} + (1-\beta) \rho_1 =\beta + (1-\beta) = 1$.
	
\smallskip
\noindent $\boldsymbol{\Leftarrow}$
For the reverse direction, suppose there exists a summary $\bS$ with $\beta \alpha_{min} + (1-\beta) \rho_1 = 1$. 
Since $\beta < 1$, it must be the case that the representativity is $\rho_1 = 1$. 
We construct an XSAT solution by setting a variable $x$ to true if $\bq_x \in \bS$ and false otherwise. We will show that the solution is an exact cover. 
Indeed, suppose that some clause $c$ is not covered by any true variable; then, $p_{\bS}(t_c) =0$, which would imply that $\rho_1 < 1$, a contradiction. Similarly, suppose that clause $c$ is covered by at least 2 variables; then, $p_{\bS}(t_c) \geq 2/m > p_{\bW}(t_c)$, which again leads to a contradiction.

\smallskip
Since the construction has only one feature $f$, we have $\alpha_{min} = \alpha_{avg} = \alpha_f$, so the same NP-hardness proof holds for $\alpha_{avg}$. Similarly, it is straightforward to see that the reduction holds for $\rho_\infty$ as well.
\end{proof}

\begin{figure*}[t]
	\begin{subfigure}{0.3\linewidth}
		\includegraphics[scale=0.65]{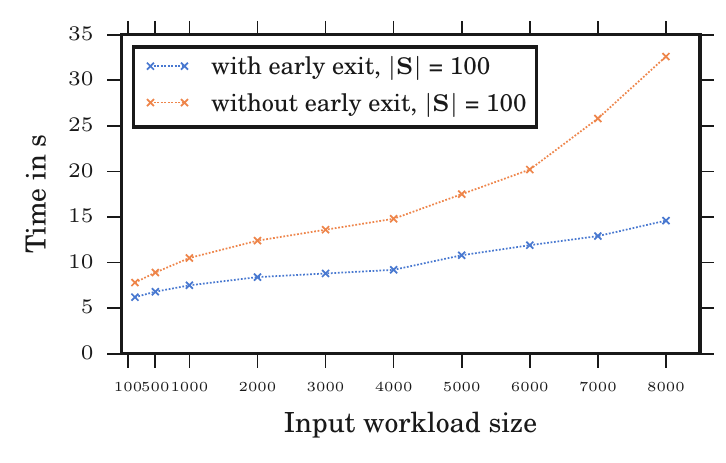}
		\caption{Runtime with varying workload size, $|\bS|=100$} \label{fig:opt1}
	\end{subfigure}
	\hspace{0.5em}
	\begin{subfigure}{0.3\linewidth}
		\includegraphics[scale=0.65]{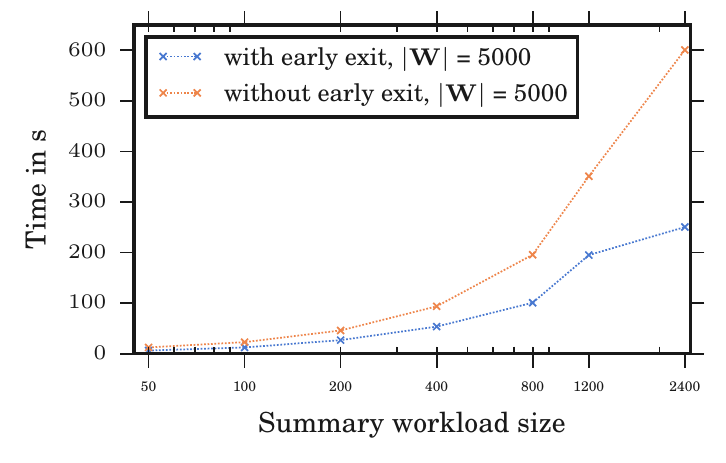} 
		\caption{Runtime with varying summary size, $|\bW|=5000$} \label{fig:opt2}
	\end{subfigure}
	\hspace{0.5em}
	\begin{subfigure}{0.3\linewidth}
		\includegraphics[scale=0.65]{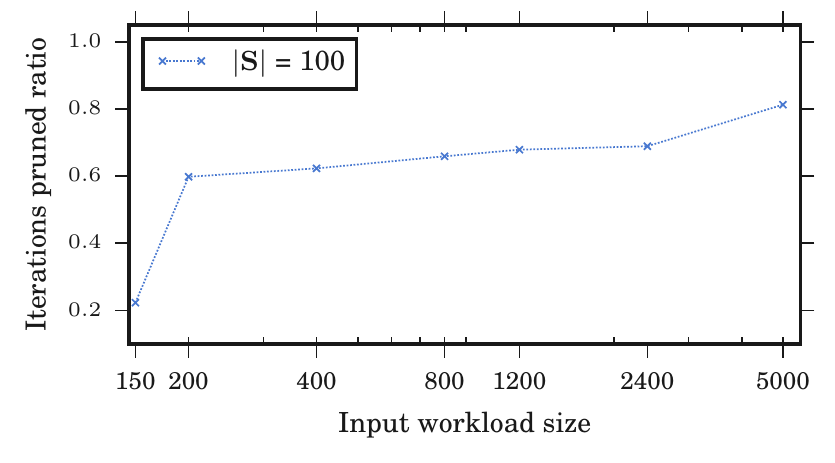} 
		\caption{Pruning ratio with varying workload size, $|S|=100$} \label{fig:opt3}
	\end{subfigure}
	
	\begin{subfigure}{0.3\linewidth}
		\includegraphics[scale=0.65]{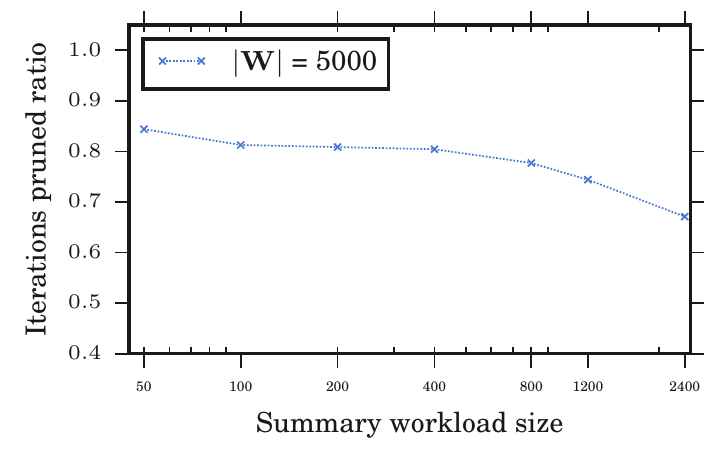} 
		\caption{Pruning ratio with varying\\summary size, $|\bW|=5000$} \label{fig:opt4}
	\end{subfigure}
\end{figure*}

\begin{proof}[Proof of Theorem~\ref{thm:hardness:apx}]

The problem is known to be APX-hard for $\beta = 1$, since it reduces to maximum coverage problem~\cite{feige1998threshold}. 
We will now show the hardness of maximizing the objective $\beta \alpha_{min} + (1-\beta) \rho_1$ for any $\beta \in [0,1)$ using a gap preserving (PTAS) reduction from the restricted \textsf{MAX-3DM} problem which we describe next.

 \smallskip
\noindent \textbf{Input.}  Three disjoint sets $A = \{a_1, \dots, a_p\}, B = \{b_1, \dots, b_p\}$ and $C = \{c_1, \dots, c_p\}$, together with a subset of triples $T \subseteq A \times B \times C$. Any element from $A,B,C$ occurs in at most 3 triples in $T$. This implies $p \leq |T| \leq 3\cdot p$.

 \smallskip
\noindent \textbf{Goal.} Find a subset $T' \subseteq T$ of maximum cardinality such that no two triples of $T'$ agree in any coordinate.
\smallskip

 Petrank~\cite{petrank1994hardness} has shown that \textsf{Max-3DM} is APX-hard even if one only allows instances where the optimal solution consists of $p = |A| = |B| = |C|$ triples; in the following we will only consider this additionally restricted version of \textsf{Max-3DM}. 
	
To achieve the PTAS reduction, we map an instance $I$ of \textsf{Max-3DM} to an instance $I'$ of the summarization problem as follows. We introduce a single multi-valued feature with domain $A \cup B \cup C$. For each triple, we create a query with feature set the values that appear in that triple. We pick the target distribution to be $\td(t) = 1/3p$ for every token $t \in A \cup B \cup C$, and use the constraint $|S| = p$.

From our instance restriction, we have $\mathsf{OPT}(I) = p$. It is also easy to observe that an optimal solution to the matching problem will produce a summarization of size $p$ where each token is covered exactly once. This implies that $\alpha_{min}=1$, $\rho_1(\td) = 1$, and hence $\mathsf{OPT}(I') = 1$.

Suppose that there exists a summary $\bS$ for $I'$ that is within $(1-\delta)$ of $\mathsf{OPT}(I')$.
Let  $c_i$ denote the number of occurrences of token $i$ in summary $\bS$. 
Further, let $V_0$ be the set of tokens with $c_i=0$, $V_1$ be the set of tokens with $c_i=1$, and $V_2$ be the remaining tokens. Denote also $n_0 = |V_0|$, $n_1= |V_1|$ and $n_2 = |V_2|$. 
Clearly we have $n_0 + n_1 + n_2  = 3p$.

Now, observe that
\begin{align*}
\rho_1 = 1 - \frac{1}{2} \sum_i |\frac{c_i}{3p} - \frac{1}{3p}| \leq 1 - \frac{1}{6p} (n_0 + n_2) 
\end{align*}

We also can write
\begin{align*}
\alpha_{min} = \frac{n_1 + n_2}{3p} = \frac{3p-n_0}{3p} \leq 1 - \frac{1}{6p} n_0
\end{align*}
Since $\beta \alpha_{min} + (1-\beta) \rho_1 \geq 1-\delta$, we have:
\begin{align*}
 \beta \cdot (1 - \frac{1}{6p} n_0) + (1-\beta) \cdot (1 - \frac{1}{6p} (n_0 + n_2) )  \geq 1-\delta
\end{align*}
from which we obtain that $n_0 + n_2 \leq 6p\delta/(1-\beta)$.

We now transform the summary $\bS$ to a matching solution $M$ for $I$ as follows: we keep only the triples for which the corresponding query covers only tokens in $V_1$. We next bound the size of the matching. Note that to cover the tokens in $V_1$ we need at least $n_1/3$ queries. Additionally, the queries that cover more than one token can be at most $3 n_2$, since every token belongs in at most 3 queries. Hence:
\begin{align*}
|M| & \geq n_1/3 - 3n_2  = (3p - n_0-n_2)/3 - 3n_2 \\
& \geq p - 10/3 (n_0 + n_2)  \geq p -20p \delta/(1-\beta) \\
& = (1-20\delta/(1-\beta)) \cdot \mathsf{OPT}(I)
\end{align*}

Choosing our $\delta$ to be $\delta(\epsilon) = \epsilon(1-\beta)/20$ completes the PTAS reduction.
To extend the hardness to the other metrics, observe that
since we have only one feature, $\alpha_{min} = \alpha_{avg}$.
It is also not hard to also show that the problem is APX-hard for $\rho_\infty$.
 \end{proof}

\begin{proof}[Proof of Lemma~\ref{lem:sum:comparison}]
We will show that there exists a value $\epsilon$ such that $G(\bS_1, \gamma) < G(\bS_2, \gamma)$ for every $\gamma \in (0, \epsilon)$. For a feature $f$, let $T_f$ denote the set of tokens that are covered by $\bS_2$ but not by $\bS_1$, and $U_f$ the set of tokens that are covered by both summaries. From the assumption for $\bS_1, \bS_2$, at least one set $T_f$ will be non-empty. Denote by $q_{min}$ the smallest probability $q(t)$ over all tokens. We can now write:
\begin{align*}
 & G(\bS_2, \gamma) -  G(\bS_1, \gamma)   =  
  \sum_f \sum_{t \in \domain(\bW, f)} q(t) \log  \left( \dfrac{\mult{\bS_2}{t}{f} +\gamma} { \mult{\bS_1}{t}{f} +\gamma} \right)  \\
  & \geq \sum_f \sum_{t \in T_f} q(t) \log  \left( \dfrac{1 +\gamma} {\gamma} \right) +
  \sum_f \sum_{t \in U_f} q(t) \log  \left( \dfrac{1} {\norm{\bW} + \gamma} \right) \\
  & \geq q_{min} \cdot \log  \left( \dfrac{1+\gamma} {\gamma} \right)  - \log  \left( \norm{\bW} + \gamma \right) 
\end{align*}
It is now easy to see that we can always pick a value $\gamma >0$ small enough to make the quantity positive.
\end{proof}

\begin{proof}[Proof of Lemma~\ref{lem:merge}]
We will begin by showing the that coverage for the merged summaries is $ \geq  \min \{\alpha, \alpha'\}/2$. Consider some feature $f$, and without loss of generality assume that $\vert \domain(\bW, f) \vert \geq \vert \domain(\bW', f) \vert$. Then, the coverage for feature $f$ for the merged summary is:
\begin{align*}
& \frac{\vert\domain(\bS, f) \; \cup \; \domain(\bS', f)\vert}{\vert \domain(\bW, f) \; \cup \; \domain(\bW', f) \vert} 
 \geq \frac{\vert\domain(\bS, f) \cup \domain(\bS', f)\vert}{2\vert\domain(\bW, f)\vert} \\
 & \geq \frac{\vert\domain(\bS, f)\vert}{2\vert\domain(\bW, f)\vert}
 = \frac{1}{2} \alpha_f 
 \geq \frac{1}{2} \alpha \geq \frac{1}{2} \min \{\alpha, \alpha' \} 
\end{align*}
	
Next, we will show that $\rho_\infty$ for the merged summary is $\geq \min\{\rho, \rho' \}$. 
From the definition of representativity, it follows that for every feature $f$ and token $t \in \domain(\bW, f)$, we have
$\vert p_{\bS}(t) - p_{\bW}(t)\vert \leq 1- \rho$ (and similarly for $\bW'$).
 %and $1 - \max_{t \in \domain(\bW_2, f)} \vert p_2(t) - q_2(t)\vert \geq \rho_2$ for every feature $f$. 
 Let $t^{\star} \in \domain(\bW \boldsymbol{\Large \biguplus} \bW', f^{\star})$ be the token for feature $f^{\star}$ that results in the largest deviation. Let $\mathit{m_{\bW}}$ be a shorthand notation for $\mult{\bW}{t^{\star}}{f^{\star}}$. 
 Then, the $\rho_\infty$ representativity for the merged summary is:
\begin{align*}
	&1 - \vert p_{\bW \boldsymbol{\Large \biguplus} \bW'}(t) - p_{\bS \boldsymbol{\Large \biguplus} \bS'}(t)  \vert  \\
	&= 1 - \bigg \vert \frac{\mathit{m_{\bW}} + \mathit{m_{\bW'}}}{\norm{\bW} + \norm{\bW'}}  - \frac{\mathit{m_{\bS}} + \mathit{m_{\bS'}}}{\norm{\bS} + \norm{\bS'}}\bigg \vert \\
	&\geq 1 - \bigg \vert \max \bigg\{ \frac{\mathit{m_{\bW}}}{\norm{\bW}}, \frac{\mathit{m_{\bW'}}}{\norm{\bW'}} \bigg\} - \min \bigg\{ \frac{\mathit{m_{\bS}}}{\norm{\bS}}, \frac{\mathit{m_{\bS'}}}{\norm{\bS'}} \bigg\}  \bigg \vert \\
	&\geq \min \{\rho, \rho'\}
\end{align*}
The first inequality follows from the assumption (without loss of generality) that $ p_{\bW \boldsymbol{\Large \biguplus} \bW'}(t) \geq p_{\bS \boldsymbol{\Large \biguplus} \bS'}(t)$ and the fact that sum of fractions are at least(at most) the smallest (largest) fraction. The second inequality follows from fixing each term with the max and min values. If $\frac{\mathit{m_{\bW}}}{\norm{\bW}} \geq \frac{\mathit{m_{\bW'}}}{\norm{\bW'}}$, then we can fix the min term to $\frac{\mathit{m_{\bW}}}{\norm{\bW}}$ since $ - \min \bigg\{ \frac{\mathit{m_{\bS}}}{\norm{\bS}}, \frac{\mathit{m_{\bS'}}}{\norm{\bS'}} \bigg\} \geq -\frac{\mathit{m_{\bS}}}{\norm{\bS}}$.
	
The cost of the merged summary is a direct consequence of the additivity of the cost function.
\end{proof}

Note that the representativity metrics can also be defined per feature. Instead of treating all tokens as the same, we can define $p_\bW(t,f) = \frac{\mult{\bW}{t}{f}} {\sum_{v \in \domain(\bW, f)}  \mult{\bW}{v}{f}}$. Given a target token distribution $d(f)$ for each feature $f$, we can now appropriately define $\rho_{1}(d(f), f)$ and $\rho_{\infty}(d(f), f)$ as in the paper and then finally set the overall $\rho_{1}$ and $\rho_{\infty}$ as the average over all the features. However, it is straightforward to show that our problem remains APX-hard even with these definitions.

\begin{proof}[Proof of Lemma~\ref{lem:merge}]
	Without loss of generality, consider a single feature $f$ with domain $\{v_1, v_2\} $. Let $\bW = \{\bq_1, \bq_2, \bq_3\}$ with feature vectors $\{ v_1\}, \{  v_2 \},$ $\{ v_1, v_2 \}$ for $\bq_1, \bq_2$ and $\bq_3$ respectively. 
	We first show that the metrics are not monotone; in other words, choosing a larger summary does not necessarily improve the representativity metric. Indeed, we have $\rho_{1}(\{ \bq_3\}) = \rho_{\infty}(\{ \bq_3\})  = 1$ but $\rho_{1}(\{\bq_1, \bq_3\}) = \rho_{\infty}(\{\bq_1, \bq_3\}) = 5/6$.
	Next, we give a counterexample to show that $\rho_{1}, \rho_\infty$ are not submodular.
	We have $\rho_{1}(\{\bq_1,\bq_3\}) + \rho_{1}(\{\bq_2,\bq_3\}) = 5/6 + 5/6$, but  $\rho_{1}(\{\bq_1,  \bq_2, \bq_3\}) + \rho_{1}(\{\bq_3\}) =  1+1 > 10/6$. 
	Exactly the same holds for $\rho_\infty$.
	%Fix $\bS_1 = \{q_3\}$ and $\bS_2 = \{q_3, q_1\}$. Then $\rho_{1}(\bS_1 \cup \bq_1) - \rho_{1}(\bS_1) = \frac{5}{6} - 1 < \rho_{1}(\bS_2 \cup \bq_1) - \rho_{L_1}(\bS_2) = 1 - \frac{5}{6}$.
\end{proof}

\begin{figure}
	\includegraphics[scale=0.50]{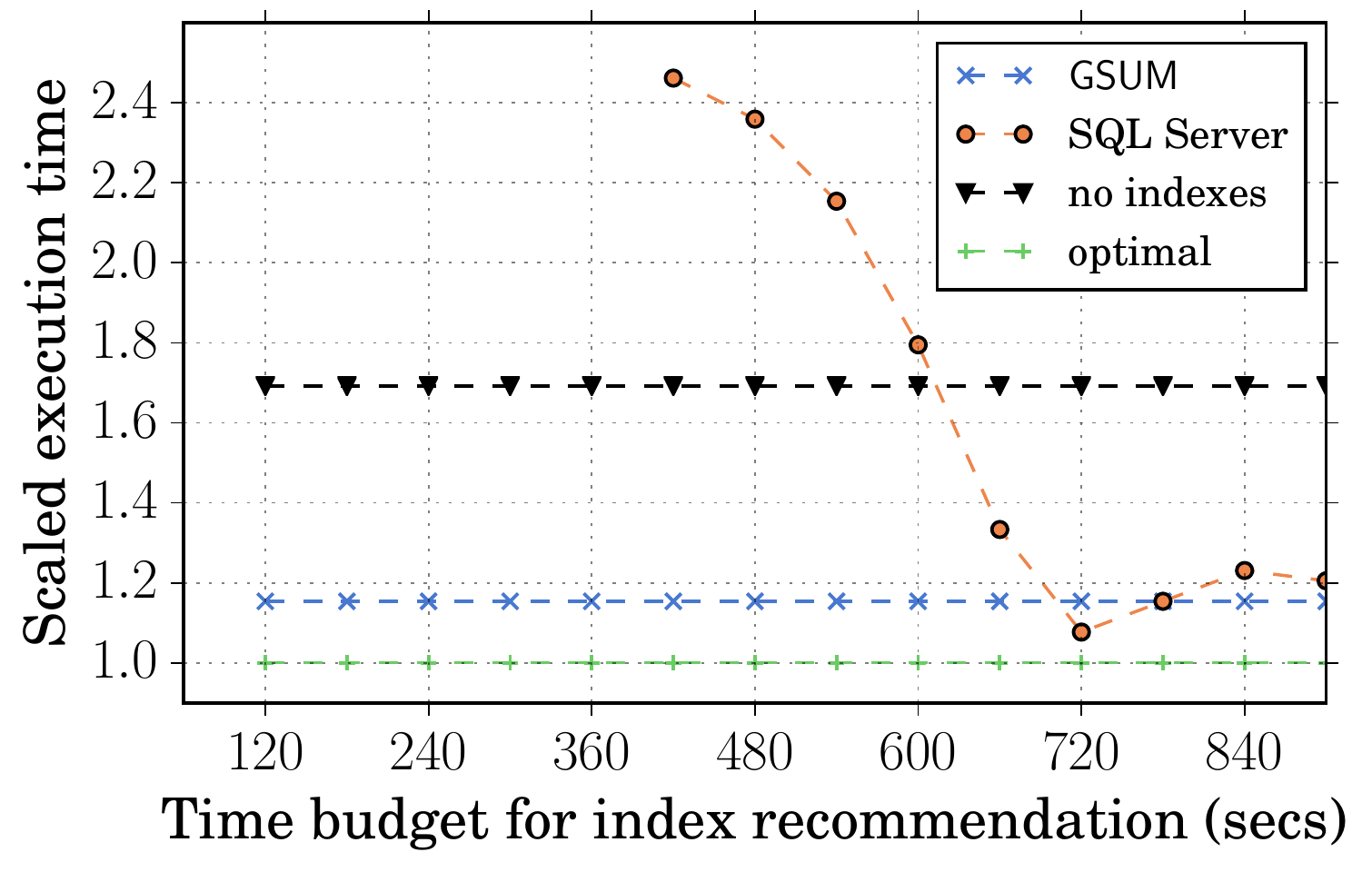}
	\caption{Index tuning time to \textsf{TPC-H skewed} workload}
	\label{fig:tpch:skew} 
\end{figure}

\begin{figure*}[t]
	\begin{subfigure}{0.3\linewidth}
		\includegraphics[scale=0.35]{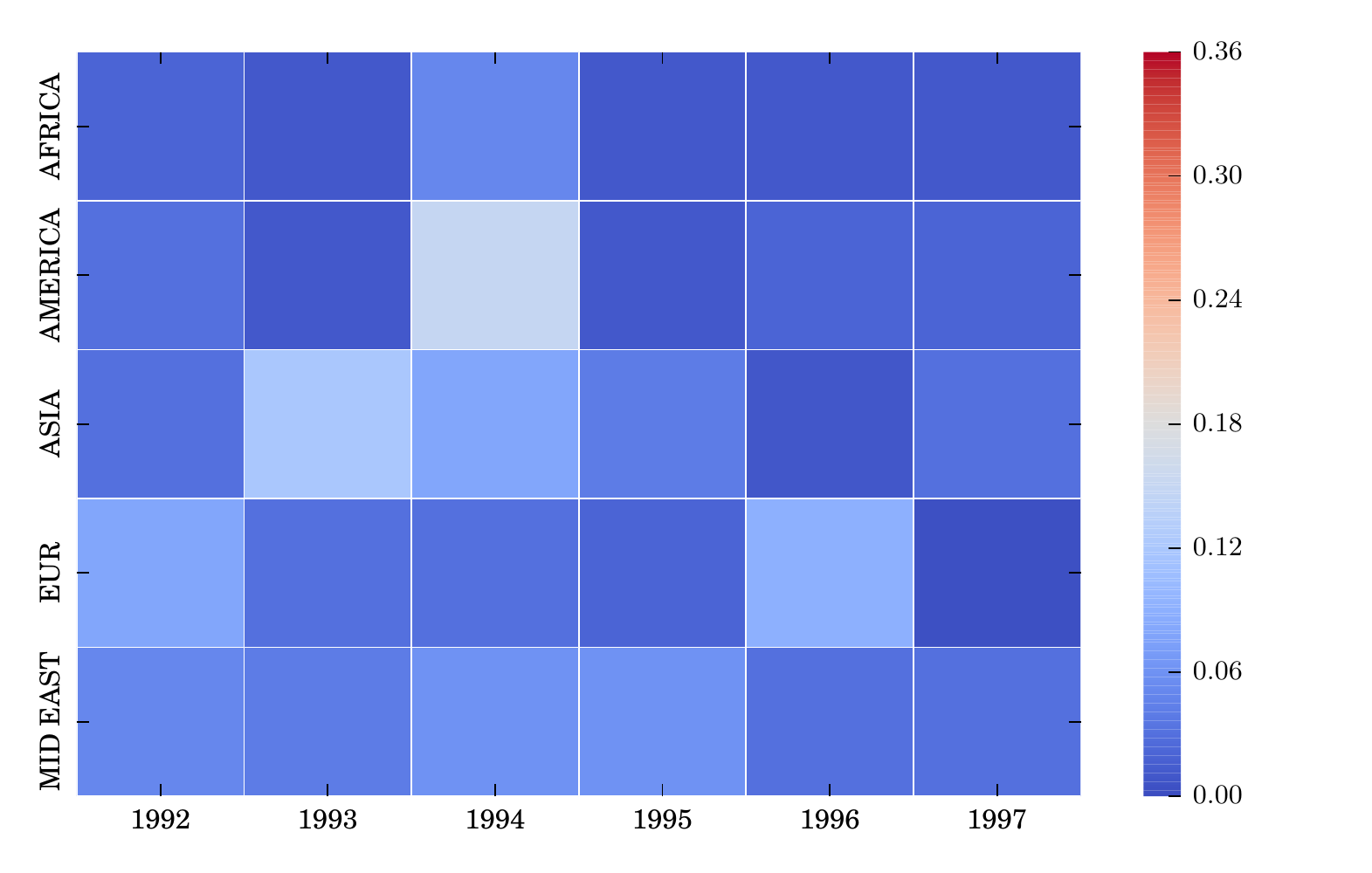}
		\caption{Input heatmap} \label{fig:heatmap:input}
	\end{subfigure}
	\hspace{0.5em}
	\begin{subfigure}{0.3\linewidth}
		\includegraphics[scale=0.35]{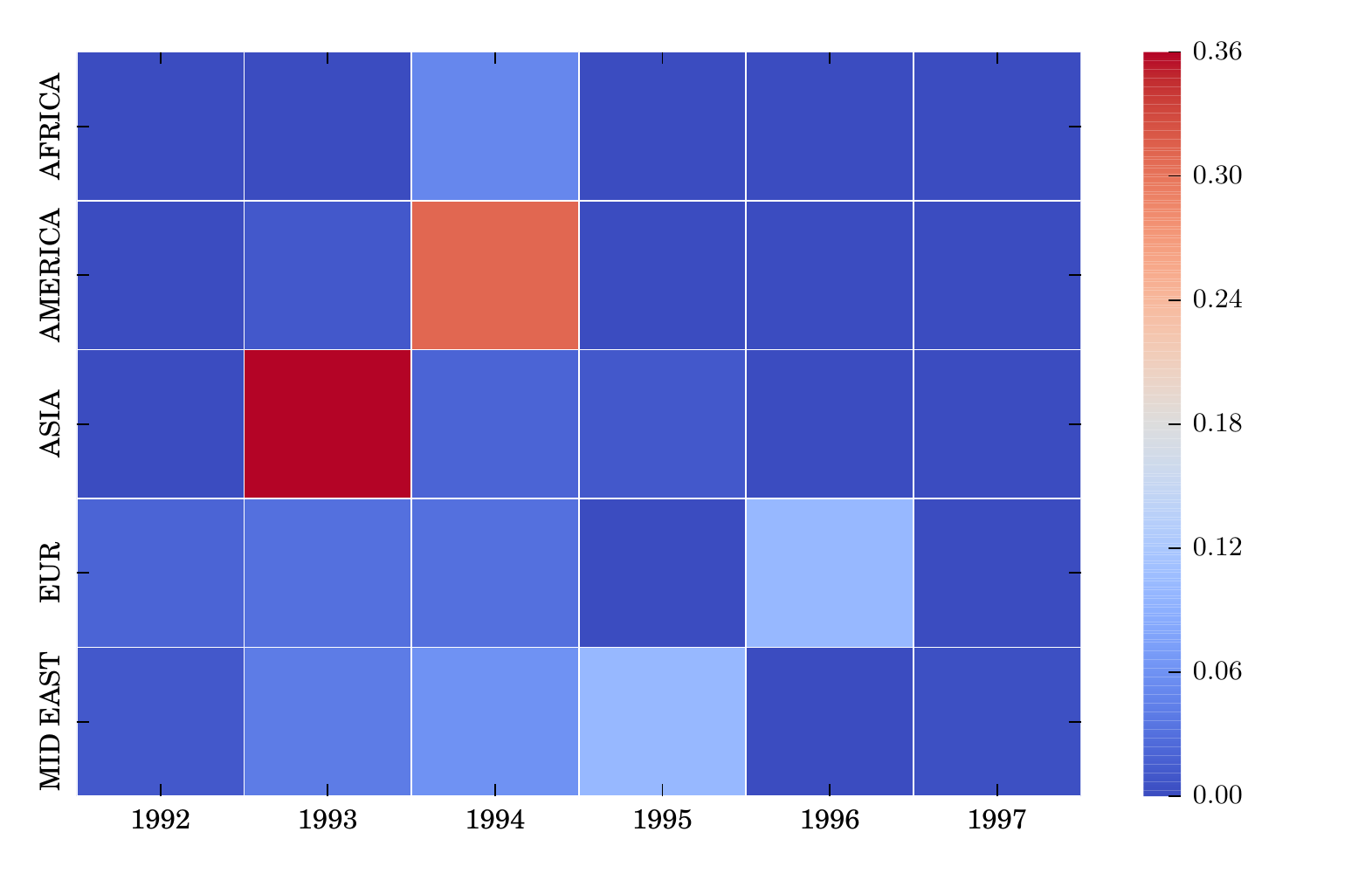} 
		\caption{Random sample heatmap} \label{fig:heatmap:randomsample}
	\end{subfigure}
	\hspace{0.5em}
	\begin{subfigure}{0.3\linewidth}
		\includegraphics[scale=0.35]{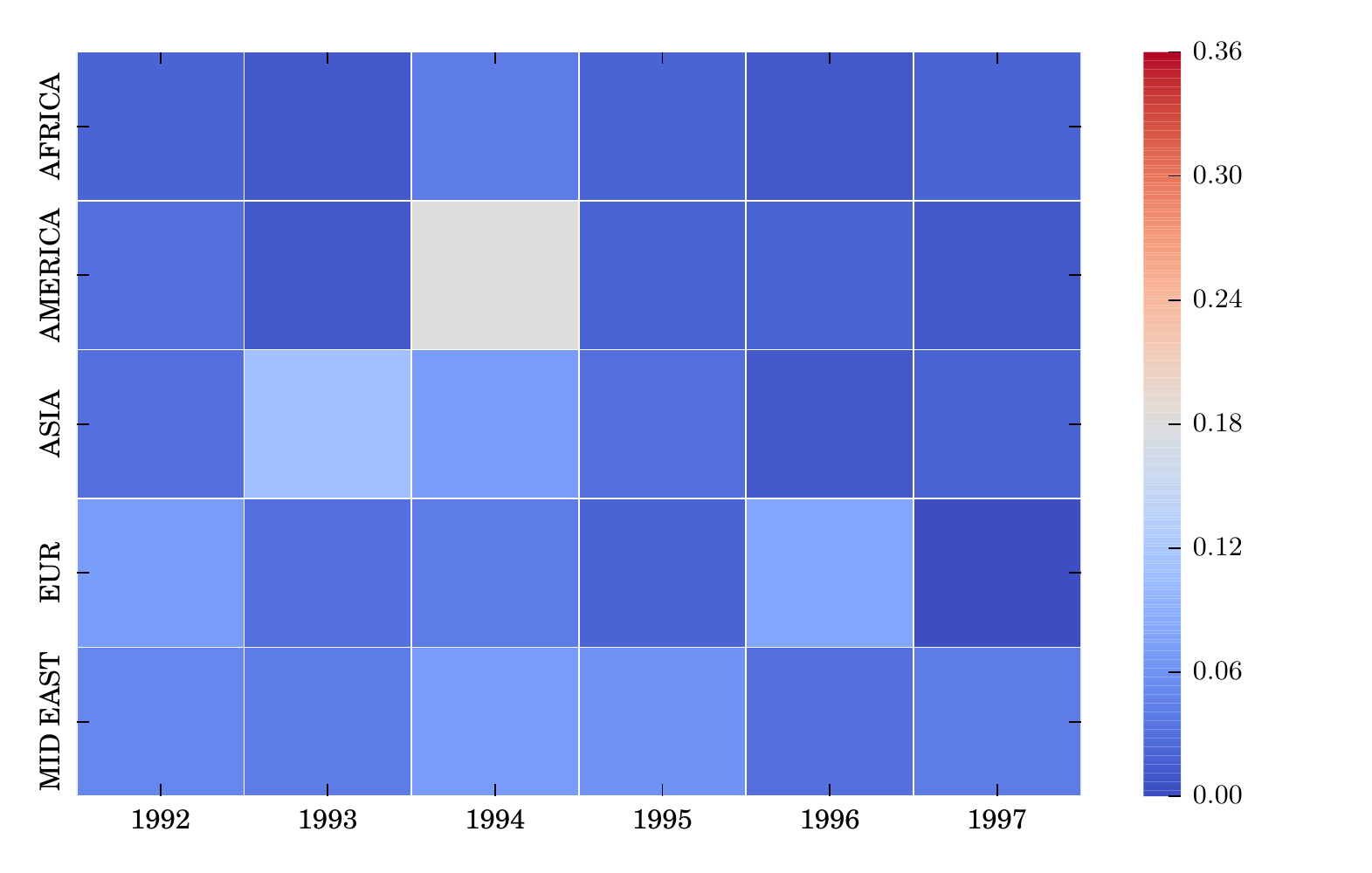} 
		\caption{Submodular  workload heatmap} \label{fig:heatmap:submodular}
	\end{subfigure}
	\caption{Heatmap for two features}
\end{figure*}

\section{Additional Experiments}

\subsection{Index Tuning}

In this section, we describe the experimental results for settings not covered in experimental section. For the first experiment, we compare random sampling with \framework.

\smallskip
\introparagraph{Results}  We study the impact of different features for compression for \textsf{TPC-H} workload. \autoref{table:features:categorical:tpch} and \autoref{table:features:numeric:tpch} show the impact for choosing one feature at a time. Similar to \textsf{SSB} workload, no one feature is good enough to extract a good summary that performs well. Even when using all numeric features, there is still a slowdown of $1.2 \times$. This is also true for categorical features. This experiment again highlights that both numeric and categorical features are important to provide enough signals so that a good summary can be extracted.

\begin{table}[t]
	\caption{Using Categorical Features: Slowdown compared to using all categorical and numeric features for compression} \label{table:features:categorical:tpch}
	\scriptsize
	\scalebox{0.90}{
		\begin{tabular}{ c|c|c|c|c}
			\toprule[0.1em]
			Task $\downarrow$ Feature $\rightarrow$  & \texttt{function\_call} &  \texttt{table\_reference} &  \texttt{group\_by} & \texttt{order\_by}  \\ \midrule[0.1em]
			\textsf{TPC-H} index & $1.56 \times$ &  $1.23 \times$ & $1.29 \times$ & $1.32 \times$  \\ \midrule
			\textsf{TPC-H} views & $1.24 \times$ &  $1.39 \times$ & $1.24 \times$ & $1.35 \times$ \\ \midrule
			\textsf{TPC-H} indexes+views & $2.2 \times$ &  $2.15 \times$ & $1.71 \times$ & $1.32 \times$ \\			
			\bottomrule
	\end{tabular}}
\end{table}

\begin{table}[t]
	\caption{Using Numeric Features: Slowdown compared to using all categorical and numeric features for compression} \label{table:features:numeric:tpch}
	\scriptsize
	\scalebox{0.90}{
		\begin{tabular}{ c|c|c|c|c}
			\toprule[0.1em]
			Task $\downarrow$ Feature $\rightarrow$  & \texttt{execution\_time} &  \texttt{output\_size} &  \texttt{\#joins} & \texttt{all numeric}  \\ \midrule[0.1em]
			\textsf{TPC-H} index & $1.2 \times$ &  $1.31 \times$ & $1.33 \times$ & $1.15 \times$\\ \midrule
			\textsf{TPC-H} views & $1.12 \times$ &  $1.27 \times$ & $1.49 \times$ & $1.09 \times$\\ \midrule
			\textsf{TPC-H} indexes+views & $1.15 \times$ &  $1.31 \times$ & $1.61 \times$ & $1.2 \times$ \\			
			\bottomrule
	\end{tabular}}
\end{table}

\subsection{Algorithmic Optimizations} 

In this section, we design experiments to test the efficiency of the {\em early exit} optimization technique applied for \framework and its effect on runtime. 
We use the same workload of $5000$ queries from Section blah. Since the baseline greedy algorithm is slow, we omit reporting its running time. 

\smallskip
\introparagraph{Results} Figure~\ref{fig:opt1} and Figure~\ref{fig:opt2} show how the two optimization techniques scale with respect to the input workload size (when summary size is kept fixed to $100$) and the summary workload size (when input workload is fixed to $5000$ queries). 
The effect of early exit optimization increases as the dataset size increases. 
This is because early exit prunes away a lot of computation in each iteration of the greedy process. 
To further investigate early pruning, let $\texttt{iter}_{\textit{total}}$ be the total number of iterations in the greedy process without any optimizations and $\texttt{iter}_{\textit{exit}}$ be the number of iterations performed by the algorithm with the early exit optimization switched on. 
Then, the pruning ratio $\texttt{iter}_{\textit{exit}}/\texttt{iter}_{\textit{total}}$ quantifies the effectiveness of the strategy. 
Figure~\ref{fig:opt3} and Figure~\ref{fig:opt4} show the impact of varying workload and summary sizes on the pruning ratio. 
As the workload size increases, the effectiveness of the optimization also increases indicating that the submodular bound becomes tighter and thus more effective. 
With respect to the summary workload size, the fraction of iterations pruned remains uniform and eventually decreases for large summaries.

\subsection{Test workload generation} \label{subsec:test:workload}

In this section, we consider the problem of generating test workloads from production queries. Developers frequently use queries to test their code while developing the functionality in \textsf{RDBMS}. For example, consider a developer who is adding the functionality of indexes on columns for faster scans. She would like to test this functionality on some queries. However, instead of choosing from a set of predefined queries, it is more desirable to choose the test workload from a set of production queries, which increases more confidence in the testing of the functionality.
%Our use-case is to generate query workloads for testing the correctness of syntax parsing and generating AST for \textsf{GoogleSQL} syntax of queries. Although developers use unit tests for this purpose, it is beneficial to use queries from production since the system will experience similar queries in the real setting. Thus, instead of choosing from a set of predefined queries, it is more desirable to choose the test workload from a set of production queries which increases more confidence in the testing of the functionality.
Additionally, we would also like to obtain metrics that tell us about the {\em coverage} of the test workload. We can now use our metrics of coverage and representativity directly here to generate a test workload from production queries for the developer. For this use-case, we create one feature for every column that appears in the \texttt{\color{blue}WHERE} clause of queries. Let $f_1, \dots, f_k$ denote all such columns. We create a combined feature $f = f_1 \times \dots \times f_k$ with domain $\domain(\bW,f) = \domain(\bW,f_1) \times \dots \times \domain(\bW,f_k)$. Intuitively, $f$ represents a $k$ dimensional cube with dimensions as the domain values of each feature and each cell in the cube stores the number of queries that {\em covers} the particular combination of values. 

\begin{example}
	Consider the \textsf{TPC-H} workload and let $f_1 = \mathsf{REGION}$ and $f_2 = \mathsf{YEAR}$. We construct a two dimensional heat map where each query can contribute to multiple cube cells if its predicates covers the values for a particular cell. Figures~\ref{fig:heatmap:input}, \ref{fig:heatmap:randomsample} and~\ref{fig:heatmap:submodular} show the heat map for the input workload, random sample of $22$ queries and a compressed workload of size $22$. The workload generated by \framework is superior to a random sample since its heat map is much more similar to the input as compared to the heat map of the random sample. Additionally, random sample has low coverage (dark blue cells) since a large number of cells are empty and it two cells (denoted in red and orange) are disproportionately represented.
	
\end{example}

\smallskip
\introparagraph{Results} We perform two sets of experiments. In the first experiment, our goal is to choose a subset workload $S \subseteq W$ that induces a data cube distribution with high coverage and high representativity (by setting $d(\cdot)$ as the input distribution). We use the \textsf{TPC-H} workload, and set the constraint total execution time of $\bS$ is at most $15$ minutes. Our algorithm is able to generate a test workload with coverage $0.75$ and representativity $0.86$ in $<10 s$. To compare with sampling, we generate multiple random samples and choose the workload with total running time at most $15$ minutes and highest metrics. The best coverage and representativity obtained is $0.17$ and $0.91$. This shows the superiority of \framework since we can achieve much higher coverage by only marginally sacrificing representativity. In the second experiment, we wish to generate a summary workload with the same constraint but set $d(\cdot)$ to be uniform instead of the input distribution. Such a scenario is useful in practice to ensure that even if the input is skewed, we generate an output workload that is skew free. Here, we obtain a summary with coverage $0.72$ and representativity $0.77$ while random sampling is unusable since $d(\cdot)$ is not the input distribution.

\subsection{More Micro-benchmarks}

\begin{figure*}[!htp]
	\begin{subfigure}{0.45\linewidth}
		\includegraphics[scale=1]{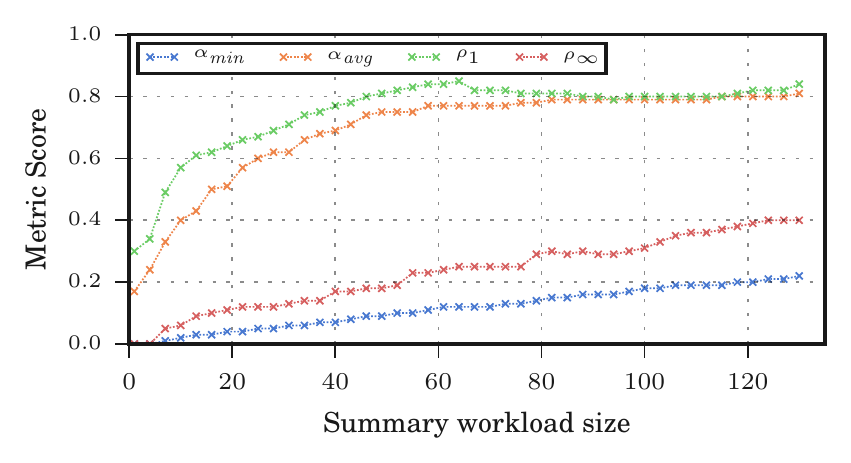}
		\caption{k-medoids on uniform TPC-H workload} \label{fig:tpch:exp1}
	\end{subfigure}
	\hspace{1cm}
	\begin{subfigure}{0.45\linewidth}
		\includegraphics[scale=1]{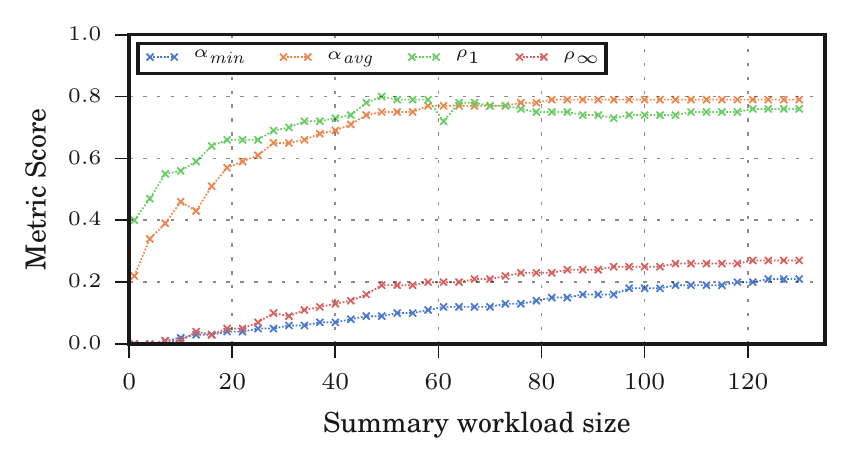}
		\caption{Submodular algorithm on uniform TPC-H workload} \label{fig:tpch:exp2}
	\end{subfigure}
	\caption{Performance metrics for uniform TPC-H workload}
\end{figure*}

In this part of the evaluation, we compare \framework with k-medoids by carefully constructing workloads where clustering performs the best and then understand the strengths and limitations of \framework.

\begin{figure*}[!htp]
	\begin{subfigure}{0.3\linewidth}
		\includegraphics[scale=0.7]{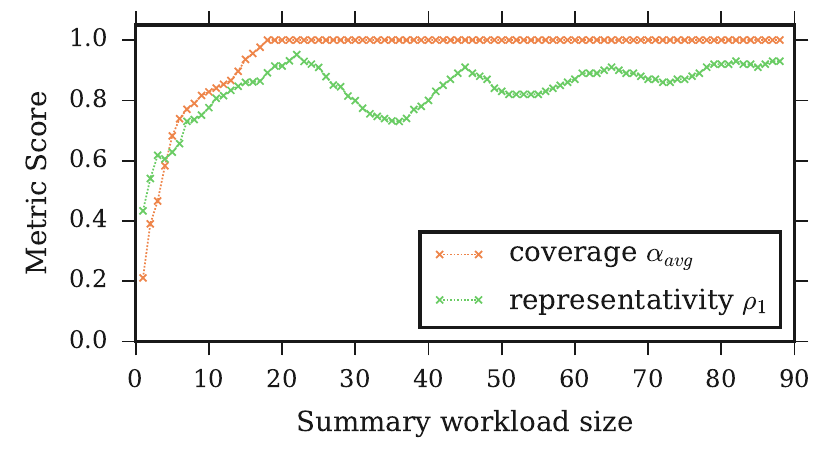}
		\caption{k-medoids} \label{fig:tpch2:exp1}
	\end{subfigure}
	\hspace{1em}
	\begin{subfigure}{0.3\linewidth}
		\includegraphics[scale=0.7]{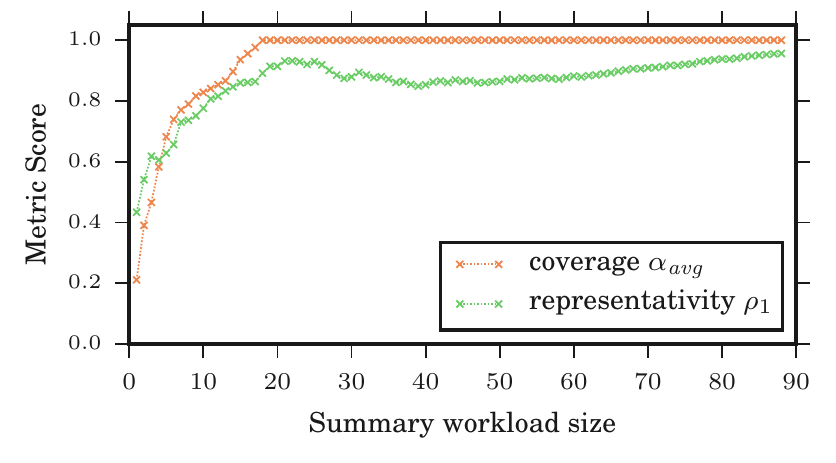}
		\caption{Submodular} \label{fig:tpch2:exp2}
	\end{subfigure}
	\hspace{1em}
	\begin{subfigure}{0.3\linewidth}
		\includegraphics[scale=0.7]{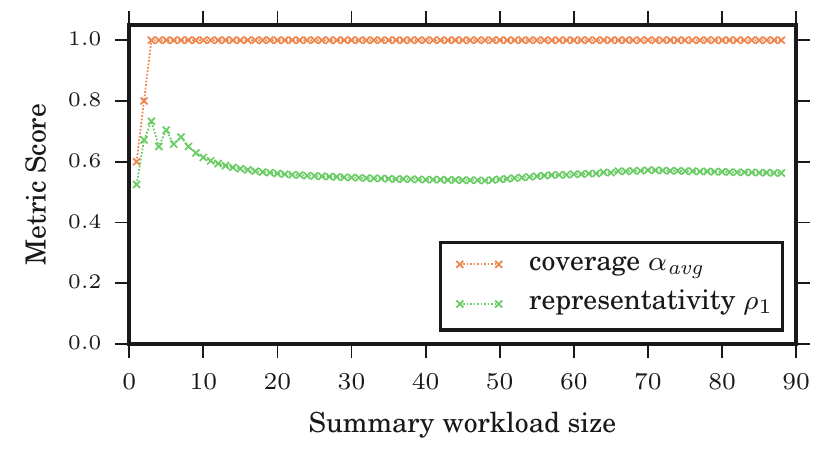}
		\caption{Skewed TPC-H workload} \label{fig:tpch2:exp3}
	\end{subfigure}
	\caption{Evaluation of the TPC-H benchmark using categorical features only.}
\end{figure*}

\begin{figure*}[t]
 	\begin{subfigure}{0.3\linewidth}
 		\includegraphics[scale=0.7]{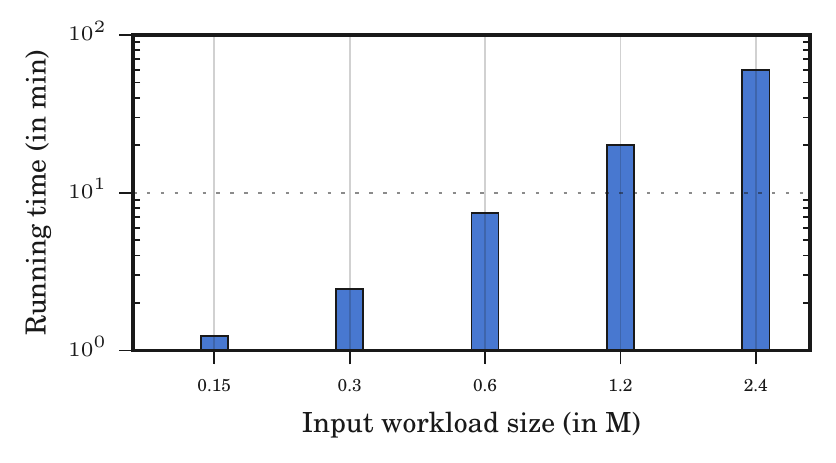}
 		\caption{Runtime varying $|\bW|$, $|\bS|=\sqrt{|\bW|}$} \label{fig:scalability:exp1:real}
 	\end{subfigure}
 \hspace{2em}
 	\begin{subfigure}{0.3\linewidth}
 		\includegraphics[scale=0.7]{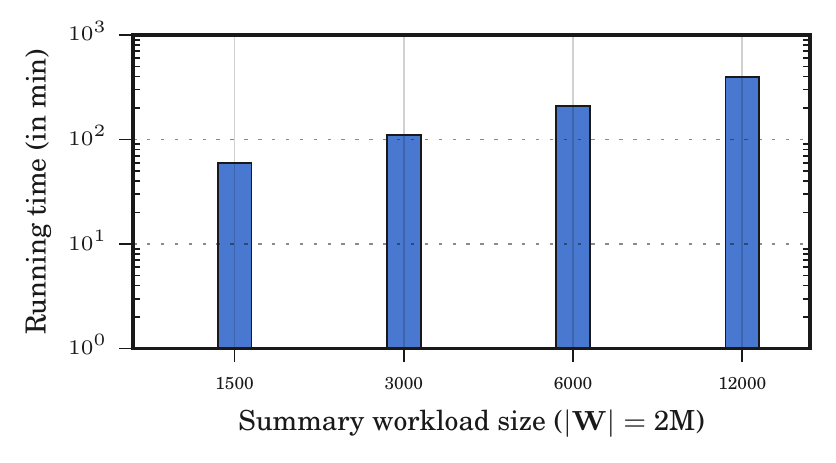}
 		\caption{Runtime varying $|\bS|$, $|\bW|=2$M} \label{fig:scalability:exp2:real}
 	\end{subfigure}
 	\caption{Scalability and Parallel Computation}
 \end{figure*} 

\smallskip
\noindent \introparagraph{Uniform Workload} For our first experiment, we construct a workload with $418$ queries where each TPC-H template is instantiated $19$ times. 
Each instantiation generates different constants in the query which affects the execution statistics. 
Then, we ask both the algorithms to generate summaries ranging from $1$ to $132$ in size. 
Our goal is observe how the different metrics change as the workload size increases. 
To create the summaries, all categorical and numeric metrics introduced in Section~\ref{subsec:encoding} are computed. Figure~\ref{fig:tpch:exp2} shows all metrics for this experiment. 
In both cases, the coverage metric shows monotonic behavior. 
The submodular algorithm has slightly better coverage than k-medoids at each step since it is trying to maximize coverage (recall that $\gamma \rightarrow 0$). However, k-medoids has a better performance for the $\rho$ metric. This behavior is expected since the initialization step of k-medoids chooses cluster center for $C_i$ by finding a query that has the largest distance from cluster center of $C_{i-1}$. 
For the TPC-H workload, all queries are {\em outliers} with respect to each other and therefore, centers are chosen perfectly. 
On the other hand, the submodular algorithm tries to maximize coverage by including at least one instance of each template until $20$. It does not pick an instance of $Q_{14}$ and $Q_{6}$ because their tokens have already been covered by other queries. 
One may find it surprising that there is not much variation in the $\rho$ metrics after $22$ and the low overall coverage. 
The main reason for this is the wide variability in runtime statistics for different instantiations of the same query template. 
To confirm this hypothesis, we calculate $\alpha$ and $\rho$ over categorical features only as shown in Figure~\ref{fig:tpch2:exp1} and Figure~\ref{fig:tpch2:exp2}. 
Here, the categorical feature coverage for both algorithms reaches $1$ in at most $22$ queries. 
Figure~\ref{fig:tpch2:exp1} demonstrates that $\rho_1$ score achieves a local maxima when the summary size is a multiple of $22$. It reaches a maximum when query size is $22$ and then marginally decreases until query $33$ before increasing again. 
The submodular algorithm behaves differently: After maximizing coverage, it picks queries according to the objective function and ensuring that queries are chosen such that the $q(\cdot)$ distribution is maintained. 
Specifically, for workload sizes between $20$ and $77$, the algorithm will always return $20$ queries as that summary size has the highest score \footnote{Recall that the greedy algorithm never picks $Q_1$ and $Q_6$ because most of its feature values are covered by other queries.}.
Beyond $77$,the $\rho_1$ score of summary size $20$ is exceeded and increases in each iteration, thus, those will be returned. 
Both algorithms eventually reach $\rho_1 = 1$ with an increase in summary size. 

So far, we have used uniform weights for all features. 
We observed in some of our experiments that the metrics can be significantly improved by assigning a higher weight to a feature that has lower $\rho$ or $\alpha$. 
We use this heuristic to perform a few weight transfer iterations by lowering the weight of features with high metrics and transferring it to features with low metrics. 
Using this technique, we were able to boost the metrics to $\rho_1 = 0.84, \rho_\infty = 0.44, \alpha_{\mathit{avg}} = 0.82, \alpha_{\mathit{min}} = 0.31$ when summary size is $22$ which is strictly better than k-medoids scores of $\rho_1 = 0.75, \rho_\infty = 0.34, \alpha_{\mathit{avg}} = 0.70, \alpha_{\mathit{min}} = 0.22$.

\smallskip
\noindent \introparagraph{Skewed Workload} Our next experiment verifies the behavior of fixing $\td(\cdot)$ to a user specified distribution. 
For this experiment, we generate a skewed TPC-H workload where query template $Q_i$ is instantiated $\lfloor \frac{1}{i+1} \cdot 66 \rfloor$ times. 
Since $Q_1$ contains \texttt{ \color{blue}SUM, MAX, COUNT}, it dominates the function calls in the workload whereas \texttt{ \color{blue}SUBSTR} and \texttt{ \color{blue} ARRAY-AGG} are present only in $Q_{22}$ and $Q_{19}$ respectively. 
To keep the setting simple, we set $\td(\cdot)$ for function call feature to be uniform, i.e,~$0.2$ for each function call token (as the domain of function call feature has $5$ values) and calculate $\rho_1, \alpha_{\mathit{avg}}$ only for function call feature. 
Figure~\ref{fig:tpch2:exp3} shows the metrics as a function of the summary workload size. 
The first query picked by \framework is $Q_1$ since it has the most function call tokens. 
The next two queries are $Q_{19}$ and $Q_{22}$ to cover the remaining two tokens. 
After another round of the same sequence, all remaining queries picked are $Q_1$ and as the summary size grows, \texttt{ \color{blue}SUM, MAX, COUNT} dominate the summary workload. 
Thus, the algorithm will return a summary size of $3$ with one instance each of $Q_1, Q_{19}, Q_{22}$ that covers all tokens and has the maximum $\rho_1$.

\section{Scalability and Parallel Computation on Production Workload} 
To benchmark scalability, we execute \framework on a single thread on a single machine and use all available categorical and numeric features. 
We use the \texttt{DataViz} workload as the basis and synthetically scale it up to $2.4M$ queries by replicating the workload.

\smallskip
\introparagraph{Results}
\Cref{fig:scalability:exp1:real} shows the runtime in minutes when the input workload size $|\bW|$ varies and the summary size is fixed to $|\bS| = \sqrt{|\bW|}$. 
If $|\bW| = 2.4$ million queries, it takes \framework $60$ minutes to execute compression. {Since the compression algorithm is executed daily, this performance is acceptable in practice. As we will see later, using multiple cores can further improve the running time as well.}
We generally observe a linear increase in runtime when increasing the workload size. \Cref{fig:scalability:exp2:real} shows how the choice of summary size impacts scalability.
Here, we set $|\bW| = 2$ million queries and observe that creating a summary with $|\bS| = 12000$ takes roughly $6$ hours.
Analogous to the results observed when increasing $|\bW|$, we see a linear increase in runtime with an increase in summary size $|\bS|$. 
 
 Next, we study the impact of the number of parallel processors used on \framework. 
 For this experiment, we take \texttt{DataViz} and scale it to a size $|\bW| = 5$M queries. Our goal is to generate of summary of $|\bS| = 2500$ ($\approx \sqrt{|\bW|}$).
 Our parallel \framework implementation partitions the input workload uniformly onto the number of available processors, $\ell$, and computes a summary of size $S$ on each processor. 
 Then, in the final step, the $|\bS| \cdot \ell$ combined queries are used to compute the output summary of size $|\bS|$. 
 
 \smallskip
 \introparagraph{Results}
 \Cref{table:runtime:production} shows the runtime and the impact of parallelization on the summary workload metrics. 
 The first column shows the metrics when a single processor calculates the summary. 
 As the number of processors increases, the speed-up obtained is near linear.
 We further observe minimal impact on coverage and representativity which only drops marginally, i.e.,~from $0.85$ to $0.81$ for coverage and $0.91$ to $0.85$ for representativity.
 
 \begin{table}[t]
 	\caption{Runtime (in minutes) and metrics obtained using parallel processing.}
 	\scalebox{0.95}{
 		\begin{small}
 			\begin{tabular}{@{}lrrrrrr@{}}\toprule
 				\textbf{\# processors $\rightarrow$} & $\mathbf{1}$&$\mathbf{2}$&$\mathbf{3}$&$\mathbf{4}$&$\mathbf{5}$&$\mathbf{6}$  \\ \toprule
 				\textbf{runtime (min)} & 180 & 100 & 64 & 45 & 38 & 31 \\ \midrule
 				\textbf{speedup obtained} & 1x& 1.8x& 2.8x & 4.0x & 4.73x & 5.8x \\\midrule
 				\textbf{coverage} & 0.85 & 0.85 & 0.84 & 0.83 & 0.82 & 0.81 \\\midrule
 				\textbf{representativity} & 0.91 & 0.88 & 0.88 & 0.86 & 0.85 & 0.85 \\
 				\bottomrule
 			\end{tabular}
 		\end{small}
 	}
 	\label{table:runtime:production}
 \end{table}
 
 \section{Feature Normalization for Incremental Computation} \label{sec:normalization}
 
 Numeric features require the knowledge of min/max values of a particular feature in order for normalization to be performed. The normalization transforms a feature value into a real number between $0$ and $1$. This is easy to perform for a given workload $\bW$ since it is easy to find our the range for all numeric features. However, when performing incremental computation, the range of a numeric feature may change as more incremental workloads are summarized. A change in the range of a numeric feature requires re-computation of the normalization value. Further, the bucket $b'_i = \lfloor v'_i \cdot H \rfloor$ may also change which requires compressing all workloads from scratch, an expensive operation.
 
 This problem can be solved by observing that the only constraint we require for normalization when performing incremental computation is that the bucket values of previously compressed workloads does not change. The key insight is that there is no technical limitation that bucket value must be between $0$ and $H$. Once the min and max value for a numeric feature is fixed, we compute the bucket value $b_i$ in the same way, except that $b_i$ can be larger than $H$. The algorithm can now update the $p_{\bW}(t)$ accordingly which is a constant time operation. Thus, as long as we have a reasonable estimate for min and max values of all numeric features, incremental computation can be performed efficiently.
 
 There are other techniques to perform feature independent normalization. Each numeric value $v_i$ can also be assigned bucket $j$ such that $v_i$ is the closest to $(1+\epsilon)^j$. The value of $\epsilon$ is chosen heuristically based on the spread of values of $v_i$. This can be done by analyzing the query log to find the appropriate value. We note that this technique is useful when most of the values are spread far apart since the bucket range increases exponentially. Instead, we can also use equi-width bucket ranges by using the function $(1+\epsilon) \cdot j$. The appropriate function is chosen on a case-by-case basis by analyzing the logs and looking at the distribution of feature values.
	
\end{document}